\documentclass{siamltex}
\usepackage{stfloats}
\usepackage{amsfonts}
\usepackage{amssymb}

\usepackage{enumitem}

\usepackage{empheq}
\usepackage{amsmath}
\usepackage{mathrsfs}
\usepackage{url}
\usepackage{subfigure}
\usepackage{graphicx}
\usepackage{cite}

\usepackage{empheq}

\usepackage{algorithm}
\usepackage{algorithmic}

\usepackage{bbm}
\usepackage{leftidx}
\usepackage{extarrows}
\usepackage{colortbl}

\newtheorem{thm}{Theorem}

\newtheorem{cor}{Corollary}
\newtheorem{lem}{Lemma}
\newtheorem{rmk}{Remark}
\newtheorem{de}{Definition}
\newtheorem{exa}{Example}
\newtheorem{pro}{Proposition}

\newcommand*{\QEDA}{\hfill\ensuremath{\blacklozenge}}

\title{Coordination on Time-Varying Antagonistic Networks}

\author{Wentao~Zhang \thanks{Tianjin Key
Laboratory of Process Measurement and Control, School of Electrical and Information Engineering, Tianjin University, Tianjin, 300072, P. R. China, {\tt wtzhangee@tju.edu.cn}} 
}
\date{}
\begin{document}

\maketitle
\begin{abstract}
This paper studies coordination problem for time-varying networks suffering from antagonistic information, quantified by scaling parameters. By such a manner, interacting property of the participating individuals and antagonistic information can be quantified in a fully decoupled perspective, thus benefiting from merely directed spanning tree hypothesis is needed, in the sense of usual algebraic graph theory. We start with rigorous argument on the existence of weighted gain, and then derive relation among weighted gain, scaling parameter and Laplacian matrix guaranteeing antagonistic information cannot diverge system state. 
Based on these arguments, we devise coordination algorithm constrained by topology-dependent average time condition, thus relaxing the examination of directed spanning tree requirement for the union graph that is usually intractable. Moreover, the induced theoretical results are applied to time-varying networks with several mutually uninfluenced agents, in accompanying with some discussions and comparisons with respect to the existing developments. 
\end{abstract}



\section{Introduction}\label{sec1}
Generally, multi-agent systems usually exhibit cooperative behavior for a collection of participating individuals, and display a unified paradigm between mathematical models and ubiquitous biology aggregation phenomena such as swarming/flocking in fish/birds. More importantly, they preserve potential for diverse practical engineering applications (cf. \cite{ren2008distributed}), typical examples include but are not limited to robotic networks, task scheduling and management, micro-scale medical treatments and smart grids/transportation, just to name a few. Up to now, fruitful results have been reported on property characterization and behavior analysis, as well as control synthesis, see monograph \cite{mesbahi2010graph} for instance. 

Unlike single-agent-level systems or lump systems,  the study of coordination problems for multi-agent systems in a changing communication environment is of great significance, apart from its own interest. Actually, from both theoretical and real-engineering viewpoints, it is always unrealistic to assume that interacting individuals enjoy a permanent interacting manner or fixed communication topology. Evidently, massive factors hold potentials to bring in time-varying interacting topology for collective behaviors. Quintessential instances contain interacting distance, obstacle circumstance, configuration and maintenance cost, as well as for the purpose of energy-saving, etc. Yet, the study of collective behavior of participating individuals within a time-varying changing setting is more challenging. This is because one is required to account for the joint effect incurred by interacting rule and communication topology, even for cooperative interaction scenarios. As a consequence, more tools are demanded to be developed, such as infinite product of stochastic matrices \cite{jadbabaie2003coordination},
Hajnal diameter \cite{hajnal1976products}, ergodicity coefficient of a stochastic matrix \cite{han2013cluster}, Lyapunov-based methods \cite{ni2010leader}, stochastic approximation \cite{huang2015stochastic}, even the geometric method \cite{qin2020on}, and so forth.

Apart from intrinsic time-varying property for participating agents, it may be more natural to take both antagonistic and cooperative interactions among participating individuals into consideration. One of sparked motivations arrives at ``friendly" interactions, to a larger sense, frequently preclude the interpretation on some interesting phenomena where both helpful and adversarial information may coexist. Typical examples include the symbiosis of friends and foes in social networks \cite{wasserman1994social,bullo2018lectures}, and the coexistence of the competition and cooperation in biofilms \cite{nadell2016spatial}, etc. Spurred by this observation, bipartite consensus problem was investigated with the help of signed graph theory and structure balance theory; see \cite{altafini2013consensus}. An extension to a directed spanning tree scenario was discussed in \cite{meng2016interval}, as well as graphical condition on stability of the agents \cite{proskurnikov2016opinion}. Further generalizations contain switching networks \cite{meng2021extended}, random signed networks \cite{shi2019dynamics}, networked multi-agent systems \cite{zhang2021velocity}, fractional-order systems \cite{gong2020practical}, and networked PDE systems \cite{chen2020bipartite}, etc.

Despite substantial progress for antagonistic networks has bee reported, a downside uncovered by mentioned literature is the joint effect incurred by interacting mechanism and description of the antagonistic information, giving rise to difficulty in examining the derived criteria. Tangible arduousness contains at least: 1) additive interacting topology is usually induced in contrast to algebraic graph theory based consensus formulation, 2) negative/positive property on weighted values of the connected edges to describe antagonistic information frequently installs huge obstacles in examining the derived conditions, even for structural balanced circumstances, and 3) as embodied by most of the existing literature on time-varying networks, directed spanning tree preserved by union graph is crucial for the underlying consensus/coordination problems. Unfortunately, such a requirement is hardly possible to check out in general \cite{chen2016convergence}. It is with above intriguing discussions in mind that this paper is concerned with coordination problem for time-varying antagonistic networks, where interacting manner among participating agents and the description of antagonistic information are quantified in a fully decoupled perspective.

Specifically, we adopt usual algebraic graph theory, sharing an inherent spirit with the classical consensus algorithm, to characterize time-varying interacting mechanism, and adopt scaling parameter to characterize whether the underlying information is hostile or not. This permits a full decoupled manner in quantifying interacting manner of agents and the involved hostile information. We then introduce heterogeneous weighted gain to assure the possibility of coordination suffering from antagonistic interactions, along with rigorous theoretical argument on its existence. To circumvent the global information that is frequently involved in deriving coordination error in the literature, we recourse to local difference based linear transformation. Along with this avenue, we devise the underlying switching law, constrained by topology-dependent average requirement, to guarantee coordination of the time-varying antagonistic networks. Moreover, some discussion and comparison with respect to the existing work are also elaborated to anchor the effectiveness for the proposed setup.

In summary, the contributions of this paper are threefold: 
\begin{enumerate}
\item [1)] We characterize interplay among agents and the description of antagonistic information in a fully decoupled manner, at which just directed spanning tree of usual algebraic graph theory is needed;
\item [2)] We give condition guaranteeing coordination of interacting agents that depends on neither the tool from infinite product of stochastic matrices, nor Lyapunov-based techniques;
\item [3)] We adopt coordination algorithm constrained by topology-dependent average time, over which there permits to assure coordination of time-varying antagonistic networks without examine directed spanning tree condition of the underlying union graph.
\end{enumerate}

The layout of this paper is outlined as follows: Section ${\rm \ref{sec2}}$ describes some basic preliminaries and the problem
formulation. Section ${\rm \ref{sec3}}$ discusses existence of the weighted gains,
the relationship among scaling parameter, weighted gain and system matrix, as well as linear transformation for coordination error. Section ${\rm \ref{sec4}}$ concentrates on coordination of interacting agents with time-vary communication topologies, in accompanying with discussion and comparison with the existing literature. Finally, a conclusion is drawn in Section ${\rm \ref{sec5}}$.

 \section{Preliminary Notation and Problem Description}\label{sec2}
\subsection{Basic Notation}
The utilized notation in this paper is standard. The nonnegative integer set, real number set and complex number set are described by the triple $(\mathbb{Z},\mathbb{R},\mathbb{C})$. Superscript $``\prime"$ represents the transpose with respect to a vector or a matrix, and ${\rm sgn}(\cdot)$ is the
sign function. Moreover, subscript (resp. superscript) $(\cdot)_{n^{2}}$ (resp. $(\cdot)^{n^{2}}$) is the abbreviation of $(\cdot)_{nn}$ (resp. $(\cdot)^{nn}$). Cardinality or module is measured by $|\cdot|$. $\| x\|$ denotes the Euclidean norm of a vector $x$, and $\|Q\|$ is the spectral norm of a matrix $Q$. $\lambda(Q)$ is the eigenvalue for a square matrix $Q$. In addition, $Q^{-1}$ is the inverse of a nonsingular matrix $Q$. $\| Q\| _{\mathscr{H}}$ refers to the spectral norm restricted in the pre-defined space $\mathscr{H}$. $\mathcal{C}_{\mathscr{F}}(\mathscr{H})$ stands for the complementary space of $\mathscr{H}$ restricted on set $\mathscr{F}$, and $\overline{\mathscr{H}}$ is the closure of space $\mathscr{H}$. Denote $\{1,...,M,M+1,...,N\}$ by $\mathbbm{I}_{N}$.

\subsection{Graph Theory}
A graph $\mathcal{G}=\{\mathbb{V},\mathbb{E},
\mathbb{A}\}$ is usually specified by a node set $\mathbb{V}$, and an edge set $\mathbb{E}$ with an adjacency matrix
$\mathbb{A}=[a_{ij}]_{N^{2}}$ satisfying $a_{ij}>0$ if $(j,i) \in \mathbb{E}$ and $a_{ij}=0$ otherwise. Self-loops are forbidden in this paper and hence $a_{ii}=0$. Laplacian matrix $\mathcal{L}=[l_{ij}]_{N^{2}}$ associated with $\mathcal{G}$ is defined by $l_{ii}=\sum^{N}_{j=1,j\neq i}a_{ij}$ and $l_{ij,j\neq i}=-a_{ij,j\neq i}$. If $(i,j) \in \mathbb{E}$, $i$ is the parent node and $j$ is the child node. A directed path connecting
nodes $v_{p_{j}}$ and $v_{p_{i}}$ is a non-repetitive sequence of edges $(v_{p_{j}},v_{p_{j}+1}), (
v_{p_{j}+1},v_{p_{j}+2}),...,(v_{p_{i}-1},v_{p_{i}})$, where $v_{p_{j}+i}\in \mathbb{V}$. A node is said to be a root in $\mathcal{G}$ if it has directed paths to every other node.
A digraph is said to be strongly connected if for any two distinct nodes,
there always exists a directed path connecting these two nodes. A digraph is said to be a directed tree if it has merely a root, and each of the remaining nodes owns exactly a parent. A directed spanning tree is a directed tree preserving all nodes in $\mathcal{G}$. A directed spanning forest is comprised of several directed trees preserving all nodes in $\mathcal{G}$.

\subsection{Problem Formulation}
Consider a family of participating agents, each of which updates its state according to
\begin{equation}\label{eq1}
\begin{aligned}
\xi_{i}(k+1)=&~\xi_{i}(k)+\epsilon u_{i}(k),~i\in\mathbbm{I}_{N},  ~k\in \mathbb{Z}
\end{aligned}
\end{equation}
where $\xi_{i}(k)\in \mathbb{R}$ is the state variable with respect to the $i$th agent, and $\epsilon>0$ is a constant. $u_{i}(k)$ represents the coordination algorithm, and features form 
\begin{subequations}\label{eqlu0}
\begin{empheq}[]{align}
&u_{i}(k)= \sum_{j\in\mathcal{N}^{r}_{i}(k)}  a_{ij}(k)(\delta_{j}\xi_{j}(k)-\delta_{i}\xi_{i}(k)),~i\in\hat{\mathbbm{I}}_{N} \label{eqlua}\\
&u_{i}(k)=\sum_{j\in\mathcal{N}^{1}_{i}(k)}a_{ij}(k)(\xi_{j}(k)-\xi_{i}(k))+\sum_{j\in\mathcal{N}^{2}_{i}(k)}a_{ij}(k)(\delta_{j}\xi_{j}(k)-\xi_{i}(k)),i\in\hat{\mathbbm{I}}^{c}_{N}  \label{eqlub}
\end{empheq}
\end{subequations}
where  $\hat{\mathbbm{I}}_{N}\triangleq\{1,...,M\}$ is the set of rooted agents and $\hat{\mathbbm{I}}^{c}_{N}=\mathbbm{I}_{N}-\hat{\mathbbm{I}}_{N}$ is the set of non-rooted agents. $a_{ij}(k)$ denotes the $(i,j)$th element in adjacent matrix $\mathbb{A}(k)$ at $k$th instant, and $\mathcal{N}^{r}_{i}(k)$ represents rooted agent $i$'s neighboring set that are also rooted agents. In addition, $\mathcal{N}^{1}_{i}(k)$ and $\mathcal{N}^{2}_{i}(k)$ are, respectively, agent $i$'s neighboring sets that are non-rooted agents and that are rooted agents satisfying $\mathcal{N}^{1}_{i}(k) \bigcap \mathcal{N}^{2}_{i}(k)=\emptyset$. Evidently, the neighboring set for non-rooted agent $i$ meets $\mathcal{N}^{1}_{i} (k)\bigcup \mathcal{N}^{2}_{i}(k)$. Scaling constant parameter $\delta_{j}\neq 0$, characterizing the antagonistic ($\delta_{j}<0$) or rewarding ($\delta_{j}>0$) information sent by the $j$th agent, is bounded for $j\in\hat{\mathbbm{I}}_{N}$. 


It is worthwhile to illustrate that merely two generic of agents are involved for multi-agent systems, i.e., rooted and non-rooted agents. It is also known that the former entirely  affects the later, not vice versa. Therefore, we specify that the interactions among rooted agents preserve both cooperative and hostile information, while these among non-rooted agents are cooperative, completely in the context of conventional multi-agent systems. This thoroughly coincides with configurations $(\ref{eqlua})$ and $(\ref{eqlub})$, respectively.

Before proceeding further, a definition regarding to the coordination problem for system $(\ref{eq1})$ is given in advance.
\begin{de}\label{de1}
Coordination problem for system $(\ref{eq1})$ is solvable if 
\begin{itemize}
\item [{\rm ({\rm \lowercase \expandafter {\romannumeral 1}})}] $\lim_{k\to \infty}\xi_{j}(k)=\lim_{k\to \infty}\delta_{i}\xi_{i}(k), \forall~i\in \hat{\mathbbm{I}}_{N}, ~j \in \hat{\mathbbm{I}}^{c}_{N}$; 
\item [{\rm ({\rm \lowercase \expandafter {\romannumeral 2}})}] $\lim_{k\to \infty}\xi_{j}(k)=\lim_{k\to \infty}\frac{\delta_{i}}{\delta_{j}}\xi_{i}(k), ~\forall j, i \in\hat{\mathbbm{I}}_{N}$; 
\item [{\rm ({\rm \lowercase \expandafter {\romannumeral 3}})}] $\lim_{k\to \infty}\xi_{j}(k)=\lim_{k\to \infty}\xi_{i}(k), ~\forall j,i \in\hat{\mathbbm{I}}^{c}_{N}$.
\end{itemize}
\end{de}

In view of Definition $\ref{de1}$, it preserves the mutually utilized definition for consensus \cite{ren2008distributed,mesbahi2010graph}, bipartite consensus \cite{altafini2013consensus}, or with respect to
scaled consensus \cite{roy2015scaled}
for some nonzero constants $\delta_{i}$ as particular cases.

As a matter of fact, the objective of this paper is to study the coordination problem for time-varying antagonistic networks in the context of usual algebraic theory (that is, $a_{ij}\geq 0$ is always permitted), which differs from the formulations in the existing literature, such as $l_{ii}=\sum^{N}_{j=1,j\neq i}|a_{ij}|$ in \cite{altafini2013consensus,meng2016interval,proskurnikov2014consensus,proskurnikov2016opinion}, or $\sum^{N}_{j=1}|a_{ij}|=1$ in \cite{meng2016behavior}. Therefore, the proposed coordination algorithm in $(\ref{eq1})$ preserves the potential to be extended to the second-order case \cite{ren2008distributed}, and the high-order case \cite{dong2017time} with minor modifications, where the communication among a family of reciprocal individuals could be adversarial. 


\section{Basic Property for Weighted Gain, Scaling Parameter and Coordination Error}\label{sec3}
\subsection{Existence of Weighted Gain}\label{sec31}
Unlike most of the existing efforts on consensus/coordination problems with antagonistic interactions, as an instance singed graph based formulation \cite{altafini2013consensus}, directed spanning tree condition suffices to assure the consensus/stability of the participating agents. By leverage of $u_{i}(k)$ in $(\ref{eqlu0})$, antagonistic information would result in the collapse of the underlying system. More precisely, let's start with a simple example.
\begin{exa}\label{exa0100}
Consider multi-agent systems $(\ref{eq1})$ with $u_{i}(k)$ in $(\ref{eqlu0})$ of form
	\begin{equation}\label{eq44}
\begin{aligned}
\xi_{1}(k+1)=\xi_{1}(k)- (\delta_{2}\xi_{2}(k)-\delta_{1}\xi_{1}(k))\\
\xi_{2}(k+1)=\xi_{2}(k)- (\delta_{1}\xi_{1}(k)-\delta_{2}\xi_{2}(k))
\end{aligned}
\end{equation}
with $\delta_{1}=1$ and $\delta_{2}=-1$. A compact expression on $(\ref{eq44})$ gives
	\begin{equation*}
\begin{aligned}
\xi(k+1)=( I-\mathcal{L}\mathcal {D})
\xi(k)
\end{aligned}
\end{equation*}
where 
\begin{equation*}
\begin{aligned}
\mathcal{L}=\begin{bmatrix}
1&-1\\
-1&1
\end{bmatrix},~\mathcal {D}=\begin{bmatrix}
1&0\\
0&-1
\end{bmatrix}.
\end{aligned}
\end{equation*}
By simple computing, one can check out that
	\begin{equation*}
	\begin{aligned}
	\mathcal{L}\mathcal {D}=\begin{bmatrix}
	1&1\\
	-1&-1
	\end{bmatrix}.
	\end{aligned}
	\end{equation*}
It is easy to see that
	\begin{equation*}
	\begin{aligned}
	{\rm det}(\lambda I-\mathcal{L}\mathcal {D})=\lambda^{2}
	\end{aligned}
	\end{equation*}
which directly suggests that matrix $\mathcal{L}\mathcal {D}$has  two zero eigenvalues, where the symbol ${\rm det}(\cdot)$ represents the matrix determinant. Evidently, coordination for multi-agent systems $(\ref{eq44})$ fails due to hostile interaction, despite of the directed spanning tree condition. \QEDA
\end{exa}

Example $\ref{exa0100}$ essentially indicates that 1) antagonistic information would be harmful for coordination of the agents; 2) directed spanning tree requirement acts merely a necessary condition for coordination of multi-agent systems, even for first-order scenario. In a nutshell, to the best of authors' knowledge, the proposed coordination algorithm in $(\ref{eqlu0})$ actually brings in something new phenomena that have not been fully discussed in the literature.

In connection with Example $\ref{exa0100}$, for the purpose of coordination for the considered system, $u_{i}(k)$ in $(\ref{eqlu0})$ is modified by
\begin{subequations}\label{eq1u}
\begin{empheq}[]{align}
&u_{i}(k)= \varrho_{i}\sum_{j\in\mathcal{N}^{l}_{i}(k)}  a_{ij}(k)(\delta_{j}\xi_{j}(k)-\delta_{i}\xi_{i}(k)),i \in \hat{\mathbbm{I}}_{N} \label{eq1ua}\\
&u_{i}(k)=\sum_{j\in\mathcal{N}^{1}_{i}(k)}a_{ij}(k)(\xi_{j}(k)-\xi_{i}(k))+\sum_{j\in\mathcal{N}^{2}_{i}(k)}a_{ij}(k)(\delta_{j}\xi_{j}(k)-\xi_{i}(k)),i\in  \hat{\mathbbm{I}}^{c}_{N}  \label{eq1ub}
	\end{empheq}
\end{subequations}
where $\varrho_{i}$ is a weighted gain whose existence shall be elaborated later.

In view of coordination algorithm $(\ref{eq1u})$, it essentially features that the interacting manner among agents and the description of antagonistic information are quantified in a fully decouple perspective.

Stacking the state variables and reformulating $(\ref{eq1})$ with $u_{i}(k)$ in $(\ref{eq1u})$ into a compact form yields
\begin{equation}\label{eq2}
\begin{aligned}
\xi(k+1)=\mathscr{L}(k)\xi(k),~k\in\mathbb{Z}
\end{aligned}
\end{equation}
with $\xi(k)=[\xi_{1}(k),...,\xi_{M}(k),\xi_{M+1}(k),...,\xi_{N}(k)]^{\prime}$, $\mathscr{L}(k)=I-\epsilon\mathcal {M}(k)$ and
\begin{equation*}
\mathcal {M}(k)=\begin{bmatrix}
\mathbb{D}\mathcal{L}_{1}(k)\mathcal {D}&\mathcal {O}_{M(N-M)}\\
\mathcal{L}_{2}(k)\mathcal {D}&\mathcal{L}_{3}(k)
\end{bmatrix},~
\mathcal {L}(k)=\begin{bmatrix}
\mathcal{L}_{1}(k)&\mathcal {O}_{M(N-M)}\\
\mathcal{L}_{2}(k)&\mathcal{L}_{3}(k)
\end{bmatrix}
\end{equation*}
where $I$ and $\mathcal {O}$ are, respectively, identity and zero matrices with appropriate
dimensions. $\mathbb{D}={\rm diag}(\varrho_{1},...,\varrho_{M})$ is the
weighted gain matrix, and $\mathcal {D}={\rm diag}(\delta_{1},...,\delta_{M})$ stands for the scaling parameter matrix. $\mathcal {L}(k)$ is the Laplacian matrix, at which matrix $\mathcal{L}_{1}(k)\in \mathbb{R}^{M^{2}}$ characterizes the interactions among rooted agents, while matrix $\mathcal{L}_{2} (k)\in \mathbb{R}^{(N-M)M}$ describes the information flows from rooted agents to associated non-rooted agents.


It is noted that matrix $\mathcal{L}_{3}(k) \in \mathbb{R}^{(N-M)^{2}}$ in $\mathcal{M}(k)$ corresponds to the
non-rooted nodes contained in communication graph $\mathcal{G}(k)$. In other words, we cannot find a
directed path whose root is associated with $\mathcal{L}_{3}(k)$ connecting some of the remaining nodes in $\mathcal{G}$.
Furthermore, it is known that the aggregated common quantity is irrelevant to the initial states of the agents
associated with $\mathcal{L}_{3}(k)$. In fact, matrix $\mathcal{L}_{3}(k)$ quantifies the information flow among non-rooted agents, apart from
in addition to the influence of rooted agents characterized by $\mathcal{L}_{2}(k)\mathcal {D}$ in matrix $\mathcal{M}(k)$.



It should be emphasized that hypothesis of the directed spanning tree is only the necessary condition to guarantee that $\mathbb{L}(k)=\mathbb{D}\mathcal{L}_{1}(k)\mathcal {D}$ has a simple zero eigenvalue. Note that even if communication graph $\mathcal{G}(k)$ attains a directed spanning tree, we still cannot declare that $\mathbb{L}(k)=\mathbb{D}\mathcal{L}_{1}(k)\mathcal {D}$ preserves a simple zero eigenvalue in general. As a matter of fact, we cannot even declare that zero eigenvalues in $\mathbb{D}\mathcal{L}_{1}(k)$ and $\mathcal{L}_{1}\mathcal {D}(k)$ is simple whenever merely directed spanning tree hypothesis in $\mathcal{G}(k)$ is assured. This is a distinguished difference in contrast to the conventional consensus problem with single dynamics (cf. \cite{ren2005consensus}). 

Fortunately, we can still establish a simple connection among the considered matrices.  
\begin{pro}\label{pro3}
	The matrices $\mathcal{L}_{1}(k)$, $\mathcal{L}_{1}(k)\mathcal {D}$, $\mathbb{D}\mathcal{L}_{1}(k)$ and $\mathbb{L}(k)$ share the same rank.
\end{pro}
\begin{proof}
	The proof is trivial using the Sylvester's rank inequality, and hence is omitted.
\end{proof}

Incorporating equation $(\ref{eq2})$ with directed spanning tree condition, it is not difficult to access that the eigenvalues of $\mathcal{L}_{3}(k)$ are in the open right half plane. Hence, the following work is to demonstrate that $\mathbb{L}(k)$ possesses a simple zero eigenvalue and the remaining eigenvalues are located in the open right half plane. This, however, is completely determined by the weighted gain matrix $\mathbb{D}$ and the scaling parameter matrix $\mathcal {D}$. To this end, the result of multiplicative inverse eigenvalue problem should be consulted in advance.
\begin{lem}{\rm(\hspace{-0.001cm}\cite[Thm. $1$]{fisher1958stabilization})}\label{le8}
For any square real matrix $\widehat{\mathbb{R}}  \in \mathbb{R}^{M^{2}}$, there exists a real diagonal matrix $\widehat{\mathbb{D}} \in \mathbb{R}^{M^{2}}$ such that the eigenvalues of the matrix $\widehat{\mathbb{D}}\widehat{\mathbb{R}}$ could be located to any desired location if all the principal minors of $\widehat{\mathbb{R}}$ are not equivalent to zero.
\end{lem}

Moreover, one also has following proposition.
\begin{pro}\label{pro4}
Suppose that graph $\mathcal{G}(k)$ attains a directed spanning tree. Then any $s$th order principal minor of $\mathcal{L}_{1}(k)\mathcal {D}$ is not equivalent to zero where $1\leq s\leq M-1$.
\end{pro}

\begin{proof}
According to \cite[Lemma $3.2$]{lin2016graph}, we immediately have access to that all $s$th order principal minors of $\mathcal{L}_{1}(k)$ are not equivalent to zero where $1\leq s\leq M-1$. The conclusion follows by virtue of the fact that the diagonal matrix $\mathcal {D}$ is nonsingular and the property of the determinant with respect to the multiplicativity for square matrices, i.e., ${\rm det}(\mathbb{R}^{\dagger}\mathbb{D}^{\dagger})={\rm det}(\mathbb{R}^{\dagger}){\rm det}(\mathbb{D}^{\dagger})$ for any square matrices $\mathbb{R}^{\dagger}$ and $\mathbb{D}^{\dagger}$.
\end{proof}

We point out here that for a general digraph with respect to $\mathcal{L}_{1}(k)$, which merely contains a directed spanning tree, the conclusion drawn by Proposition $\ref{pro4}$ may be not true.
\begin{exa}
Consider a graph with the underlying Laplacian matrix by
$\mathcal{L}_{1}=\begin{bmatrix}
	1&-1&0\\
	-1&1&0\\
	0&-1&1
\end{bmatrix}$. It is easy to check that the considered graph has a directed spanning tree. Moreover, we can also get that the $2$nd order leading principal minor of $\mathcal{L}_{1}$ is equivalent to zero. However, Proposition $\ref{pro4}$ holds if we cross out the row and column
	regarding to a root of the graph.\QEDA
\end{exa}

By using the above preparations, we will give an affirmative answer to the existence of the weighted gain $\varrho_{i}$ in updating equation $(\ref{eq1})$ with $(\ref{eq1ua})$ for $i\in \hat{\mathbbm{I}}_{N}$.
\begin{thm}\label{th4}
Suppose that graph $\mathcal{G}(k)$ attains a directed spanning tree. Then, weighted gain $\varrho_{i}$ in $(\ref{eq1})$ with $(\ref{eq1ua})$ exists for $i \in \hat{\mathbbm{I}}_{N}$ such that the eigenvalues of the matrix $\mathcal{L}_{1}(k)\mathcal {D}$ can be assigned to some desired locations other than a simple zero eigenvalue.
\end{thm}


\begin{proof}
Without loss of any generality, we assign the $M$th node being the root of the underlying graph $\mathcal{G}(k)$ (because of the specification in $(\ref{eq1u})$). Then, we partition the matrix $\mathcal{L}_{1}(k)\mathcal {D}$ by
\begin{equation*}\label{eqw10}
	\begin{aligned}
	\mathcal{L}_{1}(k)\mathcal {D}=\left[
	\begin{array}{c|c}
	\mathcal{L}_{11}(k)\mathcal {D}^{*}&    \delta_{M}\mathcal{L}_{13} (k)\\
	\hline
	\mathcal{L}_{12}(k)\mathcal {D}^{*}& \delta_{M}  l_{M^{2}}(k)\\
	\end{array}
	\right]
	\end{aligned}
\end{equation*}
where $\mathcal{L}_{1}(k)=[l_{ij}(k)]_{M^{2}}$, $\mathcal{L}_{11}(k) \in \mathbb{R}^{(M-1)^{2}}$, $\mathcal{L}_{12}(k)\in \mathbb{R}^{1(M-1)}$, $\mathcal{L}_{13}(k)\in \mathbb{R}^{(M-1)1}$ and $\mathcal {D}={\rm diag}(\mathcal {D}^{*},\delta_{M})$.
In the light of Proposition $\ref{pro4}$, all principle minors in the matrix $\mathcal{L}_{11}(k)\mathcal {D}^{*}$ are not equivalent to zero. Therefore, Lemma $\ref{le8}$ indicates that there exists such a real diagonal matrix $\mathbb{D}^{*}$ to assign the eigenvalues of $\mathbb{D}^{*}\mathcal{L}_{11}(k)\mathcal {D}^{*}$ to some desired locations, which also implies the existence of the weighted gain $\varrho_{i}$ for $i \in \hat{\mathbbm{I}}_{N}$. This completes the proof.
\end{proof}

Actually, the existence of the weighted gain $\varrho_{i}$ can always be fulfilled. As an example,
let $\varrho_{i}={\rm sgn}(\delta_{i})$ and $\delta_{j}\neq 0$ be similar to that in
\cite{roy2015scaled}, $\varrho_{i}=1$ and $\delta_{i}=1$ with $\delta_{j}={\rm sgn}(a_{ij})$ in
\cite{altafini2013consensus,proskurnikov2014consensus,proskurnikov2016opinion},
$\varrho_{i}={\rm sgn}(a_{ij})=\delta_{i}$ with $\delta_{j}=1$ in \cite{meng2016interval}
(notice that $a_{ij}\geq 0$ in $(\ref{eq1})$), or $\varrho_{i}=\mu_{i}{\rm sgn}(\delta_{i})$ for
some positive constant $\mu_{i}$, etc. Moreover, the choice of the weighted gain is far from unique
according to Proposition $\ref{pro4}$. Furthermore, $\varrho_{i}\neq 0$, where the $i$th eigenvalue
in $\mathcal{L}_{1}(k)$ is zero, should be desirable since the subgraph associated with matrix
$\mathcal{L}_{1}(k)$ is strongly connected. Evidently, $\varrho_{i}= 0$ is straightforward
if the graph $\mathcal{G}(k)$ has exactly a root $i$, since agent $i$ has no neighbors in such a circumstance.

\subsection{Connection Among Weighted Gain, Scaling Parameter and Laplacian Matrix} \label{sec32}
Although we have shown the existence of the weighted gains, it is still insufficient to thoroughly deal with the considered coordination problem. As is hinted by Example $\ref{exa0100}$, $\mathcal{L}_{1}(k)$ has a simple zero eigenvalue with directed spanning tree condition, while the matrix $\mathcal{L}_{1}(k)\mathcal {D}$ may contain multiple zero eigenvalues in addition to the nonzero eigenvalues. This is the directed consequence of multiply factor $\mathcal {D}$. However, matrix $\mathbb{L}(k)$ should merely preserve a simple zero eigenvalue and the remaining eigenvalues have positive real parts if we want to solve coordination problem depicted in $(\ref{eq1})$, which may be tightly determined by some connections among $\varrho_{i}$, $\delta_{i}$ and the principle minors of $\mathcal{L}_{1}(k)$.

\begin{figure}
	\centering
	\includegraphics[width=3in,height=2.1in]{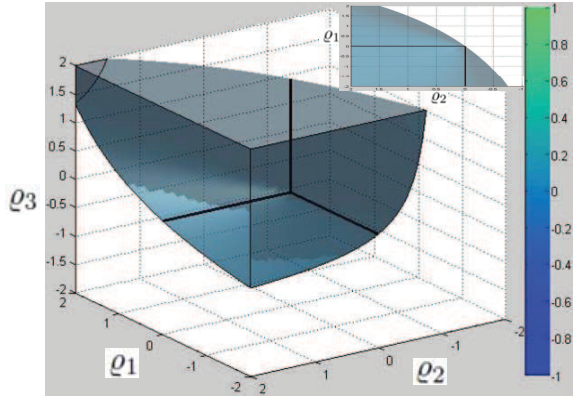}\\
	\caption{The feasible region for weighted gain $\varrho_{i}$ with the constraint $(\ref{eqw5})$.}
	\label{fig4}
\end{figure}

Let's first examine a simple example as below.
\begin{exa}\label{exa1}
Consider a communication topology with $
	\mathcal{L}_{1}=
	\begin{bmatrix}
	a&0&-a\\
	-b&b&0\\
	0&-c&c
	\end{bmatrix}$
where $a$, $b$ and $c$ are some bounded positive constants.
The characteristic equation of the underlying system with $-\mathbb{L}$ is
	\begin{equation*}\label{eqw1}
	\begin{aligned}
	{\rm det}(\lambda I+\mathbb{L})=&
	\lambda(\lambda^{2}+(\varrho_{1}\delta_{1}a+\varrho_{2}\delta_{2}b+\varrho_{3}\delta_{3}c)\lambda\\
	&+\varrho_{2}\delta_{2}b(\varrho_{1}\delta_{1}a+\varrho_{3}\delta_{3}c)+\varrho_{1}\delta_{1}\varrho_{3}\delta_{3}ac)\\
	\triangleq&\lambda(\lambda^{2}+\mathbbm{h}_{1}\lambda+\mathbbm{h}_{2})=0
	\end{aligned}
	\end{equation*}
By Routh-Hurwitz criterion, there allows a pair of $(\varrho_{i},\delta_{i})$ that does not need to share the same sign to make $\mathbb{L}$ possess a simple zero eigenvalue and two eigenvalues with positive real parts. For example, if the sign of $\varrho_{1}$ and that of $\delta_{1}$ are different, we can choose \begin{equation*}
	\begin{aligned}
	\varrho_{1}>\max_{\delta_{1}>0}\bigg\{-\frac{\varrho_{2}\delta_{2}b+\varrho_{3}\delta_{3}c}{\delta_{1}a} ,-\frac{\varrho_{2}\delta_{2}\varrho_{3}\delta_{3}bc}{\delta_{1}a(\varrho_{2}\delta_{2}b+\varrho_{3}\delta_{3}c)}  \bigg\}
	\end{aligned}
	\end{equation*}
to guarantee coordination of the considered system. Naturally, another question immediately arises. Is it possible that two pairs of $(\varrho_{i},\delta_{i})$ have different sign? The answer is NO. For the case of the $1$st and $3$rd pairs,
	\begin{equation*}
	\begin{aligned}
	\mathbbm{h}_{1}>0,~\mathbbm{h}_{2}>0~\Rightarrow~
	(\varrho_{1}\delta_{1}a+\varrho_{3}\delta_{3}c)^{2}<\varrho_{1}\delta_{1}\varrho_{3}\delta_{3}ac
	\end{aligned}
	\end{equation*}
which is a contradiction. Specifically, let
$\mathcal{L}_{1}=
\begin{bmatrix}
1&-1&0\\
-2&4&-2\\
0&-3&3
\end{bmatrix}$.
The characteristic equation of the matrix $-\mathbb{L}$ is
\begin{equation*}
\begin{aligned}
{\rm det}(\lambda I+\mathbb{L})=&\lambda(\lambda^{2}+(\varrho_{1}\delta_{1}+4\varrho_{2}\delta_{2}+3\varrho_{3}\delta_{3})\lambda\\
&+2\varrho_{1}\delta_{1}\varrho_{2}\delta_{2}+3\varrho_{1}\delta_{1}\varrho_{3}\delta_{3}+6\varrho_{2}\delta_{2}\varrho_{3}\delta_{3})=0
\end{aligned}
\end{equation*}
To assure $\mathbb{L}$ preserving two eigenvalues with positive real parts, we can determine the feasible region of weighted gain $\varrho_{i}$ by
\begin{equation}\label{eqw5}
\left\{\begin{aligned}
&\varrho_{1}<4\varrho_{2}+3\varrho_{3}\\
&\varrho_{1}(2\varrho_{2}+3\varrho_{3})<6\varrho_{2}\varrho_{3}
\end{aligned}\right.
\end{equation}
The feasible region for weighted gain $\varrho_{i}$ is presented in Fig. $\ref{fig4}$. We find that: $\textbf{1)}$ The value of $\varrho_{1}$ can be selected to be some positive constant, which indicates that ${\rm sgn}(\varrho_{1})\neq{\rm sgn}(\delta_{1})$ is desirable; $\textbf{2)}$ $\delta_{i}$ and $\varrho_{i}$ can be any nonzero real values as long as ${\rm sgn}(\varrho_{i})={\rm sgn}(\delta_{i})$ holds for each $i$.
\QEDA
\end{exa}

The following theorem describes the relationship among the signs for pairwise $(\varrho_{i},\delta_{i})$, the principle minors of $\mathcal{L}_{1}(k)$ and the configuration of the eigenvalues in the matrix $\mathbb{L}(k)$.
\begin{thm}\label{th5}
Suppose that graph $\mathcal{G}(k)$ attains a directed spanning tree. $\mathbb{L}(k)$ has a simple zero eigenvalue, and its nonzero eigenvalues are located in the open right half plane if either ${\rm sgn}(\varrho_{i})={\rm sgn}(\delta_{i})$ for all $i\in \hat{\mathbbm{I}}_{N}$, or there are at most $s$ ($1\leq s \leq M-1$) pairs of $(\varrho_{i},\delta_{i})$ having ${\rm sgn}(\varrho_{i})\neq{\rm sgn}(\delta_{i})$ such that
\begin{itemize}
\item[{\rm ({\rm \lowercase \expandafter {\romannumeral 1}})}] For each $1\leq r \leq M-1$, we need
\begin{equation*}\label{eqw2}
		\begin{aligned}
		\mathbbm{h}_{r} (k)>0
		\end{aligned}
\end{equation*}
where $\mathbbm{h}_{r}(k)=\sum\varrho_{r_{1}}\delta_{r_{1}}\cdots \varrho_{r_{r}}\delta_{r_{r}}\mathcal {J}_{r}(k)$ stands for the sum of all the $r$th order principle minors of the matrix $\mathbb{L}(k)$ with $\{ r_{1},...,r_{r}   \} \subseteq \hat{\mathbbm{I}}_{N}$. Similarly, $\mathcal {J}_{r}(k)$ is the $r$th order principle minor of the matrix $\mathcal{L}_{1}(k)$.
\item[{\rm ({\rm \lowercase \expandafter {\romannumeral 2}})}] ${\rm det}(\mathbb{H}_{r}(k))>0$ holds with $r$ ($1\leq r \leq M-1$) being all odd (or even) numbers in the set $\hat{\mathbbm{I}}_{N} \backslash \{ M\}$ where
\begin{equation*}
		\mathbb{H}_{r}(k)=
		\begin{bmatrix}
		\mathbbm{h}_{1}(k)&\mathbbm{h}_{3}(k)&\mathbbm{h}_{5}(k)&\cdots&\cdot&\cdot\\
		1&\mathbbm{h}_{2}(k)&\mathbbm{h}_{4}(k)&\cdots&\cdot&\cdot\\
		0&\mathbbm{h}_{1}(k)&\mathbbm{h}_{3}(k)&\cdots&\cdot&\cdot\\
		\vdots&\vdots&\vdots&\ddots&\vdots&\vdots\\
		0&0&0&\cdots&\mathbbm{h}_{r-2}(k)&\mathbbm{h}_{r}(k)
		\end{bmatrix}
\end{equation*}
	\end{itemize}
\end{thm}

To prove Theorem $\ref{th5}$, it necessitates an auxiliary lemma.

\begin{lem}\label{wle1}
Suppose that $\mathcal{L}(k)\in \mathbb{R}^{M^{2}}$ is the Laplacian matrix induced by a strongly connected graph. Then, any $s$th order principal minor of $\mathcal{L}(k)$ is positive where $1\leq s\leq M-1$.
\end{lem}

The proof for Lemma $\ref{wle1}$ is depicted in {\scshape Appendix} ${\rm \ref{app521}}$ for the sake of neatness.

We now dedicate to proving Theorem $\ref{th5}$.
\begin{proof}
For the first case that ${\rm sgn}(\varrho_{i})={\rm sgn}(\delta_{i})$ for all $i\in \hat{\mathbbm{I}}_{N}$, the characteristic equation of the system regarding to the matrix $-\mathbb{L}(k)$ is depicted by
\begin{equation*}\label{eqw9}
	\begin{aligned}
	\mathscr{F}_{k}(\lambda)=&~{\rm det}(\lambda I+\mathbb{L}(k))\\
	=&~{\rm det}(\lambda I+\mathcal {D}\mathbb{D}\mathcal{L}_{1}(k))=0
	\end{aligned}
\end{equation*}
It is notable that for such a scenario, the conclusion is trivial since matrix $\mathcal {D}\mathbb{D}$ is positive definite, leading to the fact that the underlying graph induced by the Laplacian matrix $\mathcal {D}\mathbb{D}\mathcal{L}_{1}(k)$ is a weighted graph \cite{agaev2005spectra}. For such a scenario, the value of $\varrho_{i}$ (resp. $\delta_{i}$) could be any bounded constant.

We now show the second case that there are some pairs of $(\varrho_{i},\delta_{i})$ with ${\rm sgn}(\varrho_{i})\neq{\rm sgn}(\delta_{i})$, and the proof is divided into two steps.
	
${\rm \lowercase \expandafter {\romannumeral 1}})$ There would permit at least a pair of $(\varrho_{i},\delta_{i})$ satisfying ${\rm sgn}(\varrho_{i})\neq{\rm sgn}(\delta_{i})$. It follows that $\mathbb{L}^{\dag}(k)=\mathcal {D}\mathbb{D}\mathcal{L}_{1}(k)\mathcal {D}\mathcal {D}^{-1}=\mathcal {D}\mathbb{D}\mathcal{L}_{1}(k)$.
By Ger$\check{\rm{s}}$gorin disk theorem \cite{horn2012matrix}, all $M$ eigenvalues in $\mathbb{L}^{\dag}(k)\triangleq[\varrho_{i}\delta_{i}l_{ij}(k)]_{M^{2}}$ are assigned in the following disks
\begin{equation}\label{eq70}
	\begin{aligned}
	\widehat{\mathscr{D}}_{i}(k)\triangleq\bigg \{ \lambda \in \mathbb{C}: |\lambda- \varrho_{i}\delta_{i}l_{i^{2}}(k)|\leq | \varrho_{i}\delta_{i}|  l_{i^{2}}(k) \bigg \}, ~i \in \hat{\mathbbm{I}}_{N}
	\end{aligned}
\end{equation}
	
We consider the worst case where each eigenvalue of $\mathbb{L}(k)$ associates with exactly a Ger$\check{\rm{s}}$gorin disk. Let $\lambda=\lambda_{1}+\mathbbm{i}\lambda_{2}$ with $\mathbbm{i}^{2}=-1$. Based on $(\ref{eq70})$, one further gets
	\begin{equation*}
	\begin{aligned}
	(\lambda_{1}-\varrho_{i}\delta_{i} l_{i^{2}}(k))^{2}+\lambda^{2}_{2}\leq (\varrho_{i}\delta_{i})^{2} l^{2}_{i^{2}}(k)
	\end{aligned}
	\end{equation*}
	It further implies that
	\begin{equation*}
	\begin{aligned}
\frac{\lambda_{1}}{l_{i^{2}}(k)}+\frac{\lambda^{2}_{2}}{\lambda_{1}l_{i^{2}}(k)}\leq 2\varrho_{i}\delta_{i}
	\end{aligned}
	\end{equation*}
where it potentially requires that $\lambda\neq 0$. Since there are $M-1$ nonzero eigenvalues, which have positive real parts, preserved in $\mathbb{L}(k)$ (resp. $\mathbb{L}^{\dag}(k)$). Therefore, one can get that there indeed exists a pair of $(\varrho_{i},\delta_{i})$ that has the opposite signs.
	
	${\rm \lowercase \expandafter {\romannumeral 2}})$ What are the conditions in the presence of such pairs to guarantee that $\mathbb{L}(k)$ has a simple zero eigenvalue while the remaining eigenvalues have positive real parts?
	
With the aid of Proposition $\ref{pro3}$, we investigate the case that matrix $\mathcal {D}\mathbb{D}$ is merely invertible by
	\begin{equation*}\label{eqw3}
	\begin{aligned}
	\mathscr{F}_{k}(\lambda)=&~{\rm det}(\lambda I+\mathbb{L}(k))\\
	=&~\lambda^{M}+\mathbbm{h}_{1}(k)\lambda^{M-1}+\mathbbm{h}_{2}(k)\lambda^{M-2}+\cdots+\mathbbm{h}_{M-1}(k)\lambda\\
	=&~\lambda(\lambda^{M-1}+\mathbbm{h}_{1}(k)\lambda^{M-2}+\mathbbm{h}_{2}(k)\lambda^{M-3}+\cdots+\mathbbm{h}_{M-1}(k))\\
	\triangleq&~\lambda \mathscr{F}^{\dag}_{k}(\lambda)=0
	\end{aligned}
	\end{equation*}
Obviously, it suffices to check out the characteristic equation $\mathscr{F}^{\dag}_{k}(\lambda)=0$. According to Lemma $\ref{wle1}$ and the Routh-Hurwitz stability criterion, $\mathscr{F}^{\dag}_{k}(\lambda)=0$ has $M-1$ roots with negative real parts if the requirements ({\rm \lowercase \expandafter {\romannumeral 1}}) and ({\rm \lowercase \expandafter {\romannumeral 2}}) hold. The conclusion hence follows.
\end{proof}

It is ready to give some comparisons in contrast to the existing literature. Since there may exist some pairs of $(\varrho_{i},\delta_{i})$ equipping with ${\rm sgn}(\varrho_{i})\neq{\rm sgn}(\delta_{i})$. Therefore, the conventional consensus problem (e.g., \cite{ren2008distributed,mesbahi2010graph}) cannot be immediately applicable to setup $(\ref{eq1})$. For instance, denoted by $z(k)=\Gamma^{\star} \xi(k)$ with $\Gamma^{\star}={\rm diag}(\mathcal {D},I)$, it yields
\begin{equation*}
\begin{aligned}
z(k+1)=(I-\epsilon\Gamma^{\star}\Gamma^{\ast} \mathcal {L}(k))z(k)
\end{aligned}
\end{equation*}
with $\Gamma^{\ast} ={\rm diag}(\mathbb{D},I)$.
Unfortunately, the fact that $\Gamma^{\star}\Gamma^{\ast}$ is not positive definite leads to that $\Gamma^{\star}\Gamma^{\ast} \mathcal {L}(k)$ is not a Laplacian matrix in the context of usual algebraic graph theory. More importantly, $I-\epsilon\Gamma^{\star}\Gamma^{\ast} \mathcal {L}(k)$ will not be a stochastic matrix in general. Thereby, the classical tool of the infinite product of stochastic matrices fails to deal with formulated problem in this paper. For the scaled consensus \cite{roy2015scaled}, there needed $\varrho_{i}={\rm sgn}(\delta_{i})$ for all $i$. Obviously, this requirement can sometimes be relaxed according to Theorem $\ref{th5}$.

\begin{rmk}
We bare the theoretical aspect consideration, Theorems $\ref{th4}$ and $\ref{th5}$ have implicitly provided a simple manner to select the underlying weighted gains, by which we can solve the coordination problem in the presence of antagonistic information. It is notable that the selection of the weighted gains could be irrelevant to the communication topology, which is not the case in \cite{lin2016graph}. An alternative is to randomly choose $\varrho_{i}$, and then we solely proceed to require that ${\rm sgn}(\varrho_{i})={\rm sgn}(\delta_{i})$, by which the matrix $\mathbb{L}(k)$ can contain a simple zero eigenvalue and the remaining eigenvalues have positive real parts. This is because that we have known the distributions of the eigenvalues regarding to $\mathcal{L}_{1}(k)$, which is a distinguished difference in contrast to \cite{lin2016graph}. Notice that it is extremely beneficial in the case where the sign of $\delta_{i}$ is prior unknown and the number of the neighbors is massive.\QEDA
\end{rmk}

\subsection{Linear Transformation for Coordination Error}\label{Main 32}
From Theorems $\ref{th4}$ and $\ref{th5}$, we are dedicated to showing that control law $(\ref{eq1u})$ is capable for coordination problem with directed spanning tree requirement. Recalling a statement suggested by Ref. \cite{ren2005consensus}, left eigenvector with respect to the zero eigenvalue of Laplacian matrix is paramount for solving the coordination problem. The result derived below indicates that if the final aggregated viewpoint for a collection of reciprocal agents is thoroughly decided by the initial conditions of $s$ $(s>1)$ agents, there exist exactly $s$ rooted nodes in the graph $\mathcal{G}(k)$. The converse argument is also true.
\begin{lem}\label{le1}
Assume that graph $\mathcal{G}(k)$ has roots. Then $\mathcal{G}(k)$ contains $s$ roots if and only if $p_{i}(k)>0$ for $i=1,...,s$, where $p(k)\triangleq [p_{1}(k),...,p_{s}(k),0,...,0]$ is the left eigenvector of Laplacian matrix associated with the zero eigenvalue, and satisfies $\sum\limits^{s}_{i=1}p_{i}(k)=1$. Moreover, for $\forall i,j \in \{1,...,s\}$, there always holds ${\rm sgn}(p_{i}(k))={\rm sgn}(p_{j}(k))$.
\end{lem}
\begin{proof}
	Lemma $\ref{le1}$ follows directly from \cite[Lemma $3.7$]{ren2005consensus}. 
\end{proof}

In terms of Lemma $\ref{le1}$, it is plausible to deal with coordination problem in $(\ref{eq1})$ by virtue of some conventional tools, such as the algebraic-based methods (continuous-time case) \cite{ren2005consensus,roy2015scaled}, the infinite product of stochastic matrices (discrete-time case)\cite{jadbabaie2003coordination,ren2005consensus,han2013cluster} in terms of the coordinate transformation on condition that both weighted gains, scaling parameters and their relationships are deterministic in advance \cite{roy2015scaled}, and the Lyapunov-based approaches \cite{ni2010leader}, just to name a few.
Unfortunately, an apparent drawback in the aforementioned references is that the final consensus interest should be previously known \cite{jadbabaie2003coordination,ren2005consensus,roy2015scaled}, or the reference variable related to a convex combination of agents' current states should be mutually computed \cite{ni2010leader}. Nevertheless, both consensus interest and reference variable are the global information for distributed systems, and they may be unavailable in general \cite{yu2015coordination}.

To overcome foregoing downside triggered by utilization of global information, a linear transformation related to coordination error is devised by
\begin{equation}\label{eq8}
\left.\begin{aligned}
&\zeta_{i}(k)=\delta_{i}\xi_{i}(k)-\delta_{i+1}\xi_{i+1}(k),~i=1,...,M-1\\
&\zeta_{M}(k)=\delta_{M}\xi_{M}(k)-\xi_{M+1}(k)\\
&\zeta_{i}(k)=\xi_{i}(k)-\xi_{i+1}(k),~i=M+1,...,N-1
\end{aligned}\right\}
\end{equation}
By leverage of linear transformation $(\ref{eq8})$, state variable $\xi(k)$ and coordination error $\zeta(k)$ enjoy a connection by
\begin{equation}\label{eq9}
\begin{aligned}
\zeta(k)=P\xi(k)
\end{aligned}
\end{equation}
where elements in matrix $P=[\mathbbm{p}_{ij}]\in \mathbb{R}^{(N-1)N}$ feature the following properties
\begin{equation*}\label{eq15}
\left.\begin{aligned}
&\mathbbm{p}_{i^{2}}=\delta_{i},~\mathbbm{p}_{i(i+1)}=-\delta_{i+1},~i\in \hat{\mathbbm{I}}_{N}\setminus \{M\}\\
&\mathbbm{p}_{M^{2}}=\delta_{M},~\mathbbm{p}_{M(M+1)}=-1\\
&\mathbbm{p}_{ij}=0,~j\neq~i,~j\neq i+1\\
&\mathbbm{p}_{i^{2}}=1,~\mathbbm{p}_{i(i+1)}=-1,~i=M+1,...,N-1
\end{aligned}\right\}
\end{equation*}


For such a transformation matrix $P$, we have the following result, which eliminates the conservative concern of solving coordination problem relying on global information, explicitly described by the induced matrix $A$.
\begin{lem}\label{le3}
For the given matrix $\mathscr{L}(k)$ in equation $(\ref{eq2})$, there exists a matrix $A(k)\in \mathbb{R}^{(N-1)^{2}}$ such that $A(k)P=P\mathscr{L}(k)$, where $A(k)=P\mathscr{L}(k)Q$ and matrix $Q =[\mathbbm{q}_{ij}]\in \mathbb{R}^{N(N-1)}$ with
\begin{equation*}
\left.\begin{aligned}
	&\mathbbm{q}_{ij}=\frac{1}{\delta_{i}},~~i\in \hat{\mathbbm{I}}_{N},j=i,...,N-1\\
	&\mathbbm{q}_{ij}=1,~~~i=M+1,...,N-1,j=i,...,N-1\\
	&\mathbbm{q}_{ij}=0, ~~~\text{ otherwise}
\end{aligned}\right\}
\end{equation*}
\end{lem}

The proof for Lemma $\ref{le3}$ is presented in {\scshape Appendix} ${\rm \ref{app5}}$.

In the sequel, the characteristics of the elements in the matrix $A(k)=[\mathbbm{a}_{ij}(k)]\in \mathbb{ R}^{(N-1)^{2}}$ could be detailed by
\begin{equation*}
\begin{aligned}
&\mathbbm{a}_{i^{2}}(k)=1+\epsilon\sum\limits^{i}_{h=1}(\varrho_{i+1}\delta_{i+1}l_{(i+1)h}(k)-\varrho_{i}\delta_{i}l_{ih}(k)), ~i< M\\
&\mathbbm{a}_{ij}(k)=0,~i\in\textbf{I}\setminus \{M\},~j=M+1,...,N\\
&\mathbbm{a}_{ij}(k)=\epsilon\sum\limits^{j}_{h=1}(\varrho_{i+1}\delta_{i+1}l_{(i+1)h}(k)-\varrho_{i}\delta_{i}l_{ih}(k)), ~j \neq i,~j\leq M\\
&\mathbbm{a}_{M^{2}}(k)=1+\epsilon\sum\limits^{M}_{h=1}(l_{(M+1)h}(k)-\varrho_{M}\delta_{M}l_{Mh}(k))\\
&\mathbbm{a}_{Mj}(k)=\epsilon\sum\limits^{j}_{h=1}(l_{(M+1)h}(k)-\varrho_{M}\delta_{M}l_{Mh}(k)),~j<M\\
&\mathbbm{a}_{Mj}(k)=\epsilon\bigg(\sum\limits^{M}_{h=1}(l_{(M+1)h}(k)-\varrho_{M}\delta_{M}l_{Mh}(k))+\sum\limits^{j}_{h=M+1}l_{(M+1)h}(k)\bigg),~j>M\\
&\mathbbm{a}_{ij}(k)=\epsilon\sum\limits^{j}_{h=1}(l_{(i+1)h}(k)-l_{ih}(k)),~i>M,~j\neq i\\
&\mathbbm{a}_{i^{2}}(k)=1+\epsilon\sum\limits^{i}_{h=1}(l_{(i+1)h}(k)-l_{ih}(k)), ~i> M
\end{aligned}
\end{equation*}

In accordance Lemma $\ref{le3}$ with linear transformation $(\ref{eq9})$, the updating equation for coordination error $\zeta(k)$ becomes
\begin{equation}\label{eq16}
\begin{aligned}
\zeta(k+1)=A(k)\zeta(k),~k\in \mathbb{Z}
\end{aligned}
\end{equation}

Despite the new induced error equation $(\ref{eq16})$ looks relatively neat, it still does not offer any clue for coordination problem released by system $(\ref{eq2})$. Thereby, one is urged to take a deep inside on the connection uncovered by matrices $\mathscr{L}(k)$ and $A(k)$, which will be presented in what follows.
\begin{lem}\label{le4}
For matrix $A(k)$ in $(\ref{eq16})$, it attains $n$ eigenvalues equaling to $1$ if and only if multiplicity of $1$ eigenvalue preserved by matrix $\mathscr{L}(k)$ is $n+1$, where $0\leq n\leq N-1$. Moreover, there are $m$ $( 0\leq m\leq N-1)$ eigenvalues of matrix $A(k)$ contained in the unit disk if and only if matrix $\mathscr{L}(k)$ has the same number of the eigenvalues within the unit disk\footnote{Here we implicitly require that ${\rm sgn}(\varrho_{i})$ and ${\rm sgn}(\delta_{i})$ for all $i \in \hat{ \mathbbm{I}}_{N}$ are in the sense of Theorem $\ref{th5}$ for the remaining sections of this paper.}.
\end{lem}

Before presenting the analysis for Lemma $\ref{le4}$, we should shed light on the relations among the eigenvalues of $\mathscr{L}(k)$ and $\mathcal {M}(k)$.
\begin{lem}\label{le5}
Suppose that conditions in Theorem $\ref{th5}$ are desirable. Then, matrix $\mathscr{L}(k)=I-\epsilon\mathcal {M}(k)$ contains a simple $1$ eigenvalue and the remaining eigenvalues are entirely contained in the unit disk.
\end{lem}

The detailed proof for Lemma $\ref{le5}$ is depicted in {\scshape Appendix} ${\rm \ref{app6}}$.

With the help of Lemma $\ref{le5}$, we are ready to complete the proof for Lemma $\ref{le4}$, and is depicted in
{\scshape Appendix} ${\rm \ref{app1}}$ for the sake of concinnity. 

Equipped with above preparations, we shall study the underlying coordination for system $(\ref{eq1})$ with control input devised in $(\ref{eq1u})$ when the interactions among participating agents are time-varying.

\section{Coordination with Time-Varying Communication Topology}\label{sec4}
\subsection{Coordination for Time-Varying Networks}
The motivations for the study of coordination problem in a changing communication setting are ubiquitous. Typical examples include higher temperature and pressure, the avoidance of unforeseen obstacles, or for the sake of adjusting the common agreement, etc. What's more important, nonlinear opinion
dynamics of social networks are often studied by reducing to a linear scenario with time-varying coupling gains \cite{proskurnikov2016opinion}.

A conventional method to deal with coordination problem in a changing environment is the infinite product of nonnegative stochastic matrices, i.e., $\prod^{\textbf{n}}_{i=1}\mathbb{S}_{i}\geq \gamma \sum^{\textbf{n}}_{i=1}\mathbb{S}_{i}$ for $\textbf{n}>1$ and $\gamma>0$, stemming from \cite[Lemma $2$]{jadbabaie2003coordination},
where $\mathbb{S}_{i}$ is a nonnegative stochastic matrix. This suggests that once graph $\sum^{\textbf{n}}_{i=1}\mathbb{S}_{i}$ has a directed spanning tree, so does graph
$\prod^{\textbf{n}}_{i=1}\mathbb{S}_{i}$, which indicates that coordination for a group of individuals can be guaranteed by Wolfowitz theorem\footnote{It should be illustrated here that Wolfowitz theorem cannot be directly implemented to the problem at hand, since $\mathscr{L}$ in $(\ref{eq2})$ is generally not a stochastic matrix anymore in the presence of antagonistic information.}. Notice also that another representative search avenue relies on Lyapunov-based techniques (see, e.g., \cite{ni2010leader}). However, both numerical
and analytical examples have suggested that it is difficult to construct a
a quadratic Lyapunov candidate for
a certain group of consensus algorithms, even for
the linear scenarios \cite{jadbabaie2003coordination,olshevsky2008no}. Furthermore, the aforementioned approaches share a common prerequisite that the union graph should possess a directed spanning tree frequently enough over time.

Evidently, this requirement is restricted, which may be
arduous to check. It is worthwhile to emphasize that infinite product of stochastic matrices may require to be checked for infinite times, giving rise to an infeasible solution \cite{chen2016convergence} on one hand, and resulting in energy or time consuming computations on the other hand. Consequently, an alternative trade-off is to impose a \emph{dwell time} restriction on each communication graph as \cite{hespanha1999stability}. It suggests the fact that the union graph contains a directed spanning tree is trivial, provided that the \emph{dwell time} constraint on each topology is satisfied. More importantly, it permits us to devise some specific switching rules along with the design procedure of the coordination algorithm \cite{zhao2012stability}. This ushers the remaining work.

Inspired by the notions of \emph{dwell time} \cite{hespanha1999stability} and \emph{mode-dependent average dwell time} \cite{zhao2012stability}, an analogous definition of \emph{topology-dependent average dwell time} (\emph{TDADT}) is presented.
\begin{de}\label{de4}
	A communication graph is said to have \emph{TDADT}-based property with respect to $\mathbb{N}_{i}$ over a set of finite sampling instants $[\mathbbm{k}_{0},\mathbbm{k}_{f}]\triangleq(\mathbbm{k}_{0},\mathbbm{k}_{0}+1,...,\mathbbm{k}_{f}-1,\mathbbm{k}_{f})$, if there exists a constant $N_{0i}\geq0$ (chatter bound \cite{hespanha1999stability} on the $i$th topology) such that
	\begin{equation*}\label{eq25}
	N^{i}_{\sigma(k)}(\mathbbm{k}_{0},\mathbbm{k}_{f})\leq N_{0i}+\frac{T_{i}(\mathbbm{k}_{0},\mathbbm{k}_{f})}{\mathbb{N}_{i}}, ~\forall~ \mathbbm{k}_{f}\geq \mathbbm{k}_{0}\geq0
	\end{equation*}
where the discrete-time function $\sigma(k)$ is defined by $\sigma(k):\mathbb{Z}\mapsto \mathscr{D}\triangleq \{1,2,...,\mathbbm{n}\}$, $1<\mathbbm{n}<\infty\in \mathbb{Z}$. $T_{i}$ ($N^{i}_{\sigma(k)}(\mathbbm{k}_{0},\mathbbm{k}_{f})$) are the total active instants (switching numbers) for the $i$th graph $\mathcal{G}_{i}$ over $[\mathbbm{k}_{0},\mathbbm{k}_{f}]$ for $i \in \mathscr{D}$.
\end{de}
\begin{rmk}
	A compelling feature on the notion of \emph{TDADT} is that it formulates the minimum execution time $\mathbb{N}_{i}$ for each potential
	communication graph over the execution interval $[\mathbbm{k}_{0},\mathbbm{k}_{f}]$. In fact,
	the union graph of graphs $\mathcal{G}_{i}$ has a directed spanning tree according to Definition $\ref{de4}$, as long as $\mathcal {T}\triangleq\sum\limits^{\mathbbm{n}}_{i=1}\mathbb{N}_{i}\leq T^{*}\triangleq
	\mathbbm{k}_{f}+1-\mathbbm{k}_{0}$ is satisfied. It may be a lower bound for the updating intervals where the union graph preserving a directed spanning tree is uniformly guaranteed.\QEDA
\end{rmk}

Now, we shall give some useful preliminaries for the theoretical support.
\begin{lem}{\rm(\hspace{-0.001cm}\cite{gohberg1986invariant})}\label{le6}
	Suppose that $\mathscr{H}_{1}$ and $\mathscr{H}_{2}$ are two subspaces of complex number space $\mathbb{C}^{\mathbbm{m}}$. The spaces $\mathscr{H}_{1}+\mathscr{H}_{2}$, $\mathscr{H}_{1}\bigcap\mathscr{H}_{2}$ and $\mathscr{H}_{1}\bigcup\mathscr{H}_{2}$ are also the subspaces in $\mathbb{C}^{\mathbbm{m}}$. For some $X\in \mathbb{C}^{\mathbbm{m}^{2}}$ and $\mathscr{H}\in \mathbb{C}^{\mathbbm{m}}$, if $\forall x\in \mathscr{H}$, $X(x)\in \mathscr{H}$ always holds, we say that $\mathscr{H}$ is an invariant space in $\mathbb{C}^{\mathbbm{m}}$.
\end{lem}

Based upon the above results, two categories of the eigenvalues are preserved in matrix $A(k)=A_{\sigma(k)}\triangleq A_{i}$ where $i\in \mathscr{D}$. That is, one category is equivalent to $1$ and the other is strictly contained in the unit disk. Additionally, notice the fact that the eigenvalues of $A_{i}$ are not equivalent to zero; and if not, there are some edges in graph $\mathcal{G}_{i}$, whose weighted values should be large enough, even may be infinite. It is of little practical usage in the design of the distributed algorithm. Thereby, we can define the following subspaces:
\begin{equation*}\label{eq26}
\begin{aligned}
\mathscr{H}^{\dag}\triangleq & \bigg\{ x\in \mathcal{C}^{\mathbbm{m}}\mid X(x)=\lambda x, \lambda=1 \bigg\},~
\mathscr{H}^{\ddag}\triangleq  \bigg\{ x\in \mathcal{C}^{\mathbbm{m}}\mid X(x)=\lambda x, | \lambda|<1\bigg \}
\end{aligned}
\end{equation*}

Apparently, by Lemma $\ref{le6}$, both $\mathcal {H}^{\dag}$ and $\mathcal {H}^{\ddag}$ are the invariant
subspaces of the $\mathcal{C}^{\mathbbm{m}}$. We hence derive the result below.

\begin{lem}\label{le7}
Consider equation $(\ref{eq16})$. Letting $\lambda=\max\{|\lambda(A(k))|: | \lambda(A(k))|<1\}$, it follows that $
	\| A(k) \|_{\mathscr{H}^{\dag}}\leq\rho$ and $\| A(k) \|_{ \mathscr{H}^{\ddag}}\leq \rho \lambda$
	where $\rho$ is a positive scalar.
\end{lem}

The detailed proof for Lemma $\ref{le7}$ is depicted in {\scshape Appendix} ${\rm \ref{app7}}$.

Subsequently, define $\lambda^{i}\triangleq \max \{  |\lambda(A_{i})| <1\}$ for $ i \in \mathscr{D}$. By Lemma $\ref{le7}$, it follows that
\begin{equation}\label{eq31}
\begin{aligned}
\| A_{i}\| _{\mathscr{H}_{i}^{\dag}}\leq&~ \| U\|_{\mathscr{H}_{i}^{\dag}} \| U ^{-1}\|_{\mathscr{H}_{i}^{\dag}}\leq \rho^{i}\\
\| A_{i}\|_{\mathscr{H}_{i}^{\ddag}}\leq&~ \|U\|_{\mathscr{H}_{i}^{\ddag}}\| U ^{-1}\|_{\mathscr{H}_{i} ^{\ddag}}\|\textbf{A}_{i}\|\leq \rho^{i} \lambda^{i}
\end{aligned} ~~~ \text{$i \in \mathscr{D}$}
\end{equation}


\begin{lem}\label{le9}
	The union graph of a group of changing topologies contains a directed spanning tree if there exist two sets $\mathcal{S}_{1}$ and $\mathcal{S}_{2}$, both of which belong to $\mathscr{D}$ and satisfy $\mathcal{S}_{1}\bigcup\mathcal{S}_{2}=\mathscr{D}$, such that $\mathscr{C}_{1}=\sum\limits_{i\in\mathcal{S}_{1}}\mathscr{H}_{i}^{\dag}$ and $\mathscr{C}_{2}=\bigcap\limits_{ i\in \mathcal{S}_{1}}  \mathscr{H}_{i}^{\ddag}$ are $A_{i}$-invariant sets, and
	\begin{equation}\label{eq36}
	\begin{aligned}
	\mathscr{C}_{1 }\subseteq\bigcap\limits_{i\in\mathcal{S}_{2}}\mathscr{H}_{i}^{\ddag}
	\end{aligned}
	\end{equation}
\end{lem}

Before presenting the proof of Lemma $\ref{le9}$, we need a result borrowed from \cite{ni2010leader}. 
It actually manifests that the union graph is jointly connected if the $(N-1)$ nodes have received information from their neighbors, and each eigenvalue in Laplacian matrix corresponding to one of $(N-1)$ nodes is not equivalent to zero over a uniformly bounded interval.
\begin{lem}{\rm(\hspace{-0.001cm}\cite{ni2010leader})}\label{le12}
	A class of potential changing topologies (each of them has $N$ nodes) over $[\mathbbm{k}_{0},\mathbbm{k}_{f}]$ are jointly connected, if
	\begin{equation}\label{eq71}
	\bigcup\limits_{k\in[ \mathbbm{k}_{0},\mathbbm{k}_{f}]}\Phi(\sigma(k))=\bigg\{ 1,...,N-1 \bigg \}
	\end{equation}
	with $\Phi(i)=\{s\mid \lambda^{s}_{i}\neq 0, s=1,...,N-1, i\in    \mathscr{D}\}$, where $\lambda^{s}_{i}$ is the $s$th eigenvalue of the $i$th Laplacian matrix.
\end{lem}

Consensus problems for multi-agent systems with changing topologies have been extensively studied in the literature, see \cite{jadbabaie2003coordination,ren2005consensus} and references therein. Obviously, a directed spanning tree preserved in the union graph of a finite family of graphs over a uniformly bounded interval is of great importance for solving the involved consensus problems with dynamically changing communication topologies. Lemma $\ref{le12}$ explicitly elaborates the rationale behind the switching communication of a collection of participating agents. Furthermore, it is easy to examine that the result developed in Lemma $\ref{le12}$ is also capable in the case of a directed graph.


We now give a detailed proof for Lemma $\ref{le9}$, see {\scshape Appendix} ${\rm \ref{app8}}$ for details.

\begin{figure}
	\centering
	\includegraphics[width=2.9in,height=1.3in]{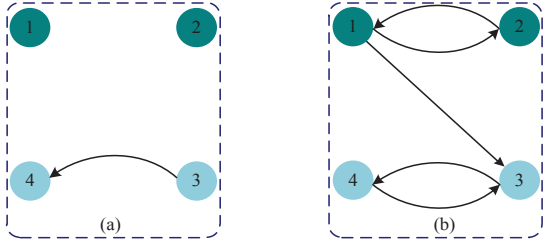}\\
	\caption{A schematic paradigm for the underlying switching topologies where the upper nodes act as rooted agents.}
	\label{fig2}
\end{figure}

Lemma $\ref{le9}$ in fact suggests that the union graph of a group of switching graphs preserves a directed spanning tree if
the requirements in Lemma $\ref{le9}$ are satisfied. However, the converse is generally not true.


Here we present an example to show the fact that the conditions in Lemma $\ref{le9}$ may be less conservative to some extent.


\begin{exa}\label{exa3}
	Assume that the underlying switching graphs are given in Fig. $\ref{fig2}$. The scaling and weighted parameters in $(\ref{eq1u})$ are
	selected as $\delta_{1}=0.75$, $\delta_{2}=-0.85$, $\varrho_{1}=0.75$, $\varrho_{2}=-0.8$ and $\epsilon=0.15$. By computation, it acquires
	\begin{equation*}\label{61}
	A_{1}=\begin{bmatrix}
	1&0&0\\
	0 &1&0\\
	0&0&0.85
	\end{bmatrix},~
	A_{2}=\begin{bmatrix}
	0.8136&0&0\\
	-0.048&0.85&0.15\\
	0.15&0.15&0.7
	\end{bmatrix}
	\end{equation*}
	with $a_{ij}=1$ or $a_{ij}=0$ and their corresponding eigenvalues are $\lambda(A_{1})=(1,1,0.85)$ and $\lambda(A_{2})=(0.6073,0.9427,0.8136)$. It is easy to check that condition $(\ref{eq71})$ holds. Denote $\mathcal{S}_{1}=\{{\rm (a)}\}$ and $\mathcal{S}_{2}=\{{\rm (b)}\}$ for graphs depicted in Fig. $\ref{fig2}$. Moreover, it is easy to see that $\mathscr{C}_{1}$ and $\widetilde{\mathscr{C}}_{2}\triangleq \bigcap\limits_{i\in\mathcal{S}_{2}}\mathscr{H}_{i}^{\ddag}$ are, respectively, spanned by
	\begin{equation*}
	\left\{\begin{aligned}
	\begin{bmatrix}
	1\\
	0\\
	0
	\end{bmatrix}, \begin{bmatrix}
	0\\
	1\\
	0
	\end{bmatrix}
	\end{aligned}\right\},\\
	\left\{\begin{aligned}
	\begin{bmatrix}
	0    \\
	0.5257  \\
	-0.8507   \\
	\end{bmatrix},
	\begin{bmatrix}
	0   \\
	-0.8507  \\
	-0.5257  \\
	\end{bmatrix},
	\begin{bmatrix}
	0.782\\
	-0.5005\\
	0.3716\\
	\end{bmatrix}
	\end{aligned}\right\}
	\end{equation*}
	It is prone to see that $(\ref{eq36})$ is desirable.\QEDA
\end{exa}

\begin{rmk}\label{rmk5}
	A potential approach ensuring $(\ref{eq36})$ in deploying a network of agents under the switching setting is to guarantee that the whole network possesses a directed spanning tree during a period of execution time. A toy instance to motivate the aforementioned method is the fire surveillance for a region of forest in terms of a group of sensors. Obviously, only parts of the sensors are needed to be woken up during the rainy seasons, while all the sensors should be woken up during the dry seasons.\QEDA
\end{rmk}

It is now ready to present the main result of this paper with a changing interaction setting as below.

\begin{thm}\label{th2}
For given scalars $\omega^{i}>0$ and $\gamma^{i}>1$. Suppose that the requirements in Lemma $\ref{le9}$ are satisfied.Then, coordination problem for system $(\ref{eq1})$ with $(\ref{eq1u})$ can be solved under the changing communication environment, if the switching signal $\sigma(k)$ satisfies
\begin{subequations}  
	\begin{align}  
&	\omega^{i}\geq \sqrt[\mathbb{N}_{i}]{\rho^{i}}, ~i\in \mathscr{D}      \label{eq32} \\
&	\lambda^{i}\gamma^{i}\leq \frac{1}{\sqrt[\mathbb{N}_{i}]{\rho^{i}}}, ~i\in \mathscr{D}	      \label{eq33}\\
&	\prod\limits_{i\in\mathcal{S}_{1}}(\omega^{i})^{T_{i}(0,K)} \prod\limits_{i\in\mathcal{S}_{2}}(\frac{1}{\gamma^{i}})^{T_{i}(0,K)}
\leq\prod\limits_{i\in\mathscr{D}}(\delta^{i})^{T_{i}(0,K)}	\label{eq34}\\
&\prod\limits_{i\in\mathcal{S}_{1}}(\frac{1}{\gamma^{i}})^{T_{i}(0,K)}
\prod\limits_{i\in\mathcal{S}_{2}}(\omega^{i})^{T_{i}(0,K)} \leq\prod\limits_{i\in\mathscr{D}}(\delta^{i})^{T_{i}(0,K)}\label{eq35}
	\end{align}
\end{subequations}
where $0<\delta^{i}<1$ and $T_{i}$ is confined in Definition $\ref{de4}$.
\end{thm}

\begin{proof}
Actually, Lemmas $\ref{le6}$ and $\ref{le7}$ indicate the underlying facts: $\mathscr{C}_{i}$ is a subspace of $\mathscr{C}^{N}$ for $i=1,2$, $\bigcap\limits^{2}_{i=1}\mathscr{C}_{i}=\emptyset$. Besides, one also gets that $\mathcal{C}_{\mathscr{C}^{N}}(\mathscr{C}_{1})=\mathscr{C}_{2}$.
	
We proceed with denoting the switching instances over the sampled interval $[0,K]$ by $0<k_{1},...,k_{N_{\sigma}(0,K)}$ for any sufficient large positive integer $K$ where $N_{\sigma}(0,K)=\sum\limits^{\mathbbm{n}}_{i=1}N_{i \sigma}(0,K)$. For any initial value $\zeta(0)\in \mathscr{C}_{1}$, it has
\begin{equation}\label{eq37}
	\begin{aligned}
	\zeta(K+1)=&(A_{\sigma(k_{N_{\sigma}(0,K)})})^{\hat{k}_{N_{\sigma}(0,K)}}\cdots(A_{\sigma(k_{i+1})})^{\hat{k}_{i}}\cdots (A_{\sigma(k_{1})})^{\hat{k}_{0}}\zeta(0) \in \mathscr{C}_{1}
	\end{aligned}
\end{equation}
where $\hat{k}_{N_{\sigma}(0,K)}=K+1-k_{N_{\sigma}(0,K)},...,\hat{k}_{i}=k_{i+1}-k_{i},...,\hat{k}_{0}=k_{1}-k_{0}$ with $k_{0}=0$.
	
	Denote two sets $\Omega_{1}$ and $\Omega_{2}$ such that $\sigma_{k_{i}}$ is pertained to $\mathcal{S}_{1}$ and $\mathcal{S}_{2}$ for $i\in \{ 1,..., k_{N_{\sigma}(0,K)}\}$. According to $(\ref{eq31})$, $(\ref{eq36})$ and Definition $\ref{de4}$, one has
	\begin{equation*}\label{eq38}
	\begin{aligned}
	\|\zeta(K+1)\|
	\leq&\prod\limits_{i\in \Omega_{1}}\| (A_{\sigma(k_{i+1})})^{\hat{k}_{i}}\|_{\mathscr{C}_{1}} \prod\limits_{i\in \Omega_{2}}\| (A_{\sigma(k_{i+1})})^{\hat{k}_{i}}\|_{\mathscr{C}_{1}}
	\|\zeta(0)\|\\
	\leq&\prod\limits_{i\in \Omega_{1}}\| (A_{\sigma(k_{i+1})})^{\hat{k}_{i}}\|_{\mathscr{C}_{1}} \prod\limits_{i\in \Omega_{2}}\| (A_{\sigma(k_{i+1})})^{\hat{k}_{i}}\|_{\mathscr{H}_{i}^{\ddag}}
	\|\zeta(0)\|\\
	\leq&\prod\limits_{i\in\mathscr{D}}(\rho^{i})^{N_{0i}} \prod\limits_{i\in\mathcal{S}_{1}} (\rho^{i})^{\frac{T_{i}(0,K)}{\mathbb{N}_{i}}}\prod\limits_{i\in\mathcal{S}_{2}} (\rho^{i}(\lambda^{i})^{\mathbb{N}_{i}})^{\frac{T_{i}(0,K)}{\mathbb{N}_{i}}} \|\zeta(0)\|\\
	=&\prod\limits_{i\in\mathscr{D}}(\rho^{i})^{N_{0i}} \prod\limits_{i\in\mathcal{S}_{1}} (\frac{\rho^{i}}{(\omega^{i})^{\mathbb{N}_{i}}})^{\frac{T_{i}(0,K)}{\mathbb{N}_{i}}}\prod\limits_{i\in\mathcal{S}_{2}} (\rho^{i}(\lambda^{i}\gamma^{i})^{\mathbb{N}_{i}})^{\frac{T_{i}(0,K)}{\mathbb{N}_{i}}} \times\\
	&  \times \prod\limits_{i\in\mathcal{S}_{1}} (\omega^{i})^{T_{i}(0,K)}\prod\limits_{i\in\mathcal{S}_{2}} (\frac{1}{\gamma^{i}})^{T_{i}(0,K)} \|\zeta(0)\|.\\
	\end{aligned}
	\end{equation*}
	Specifying $(\ref{eq32})$ and $(\ref{eq33})$ by
	$
	\frac{\rho^{i}}{(\omega^{i})^{\mathbb{N}_{i}}}\leq 1, ~i\in \mathcal{S}_{1}$,
	$\rho^{i}(\lambda^{i}\gamma^{i})^{\mathbb{N}_{i}}\leq 1, ~i\in \mathcal{S}_{2}$.
	It follows from $(\ref{eq34})$ that
	\begin{equation*}
	\begin{aligned}
	\|\zeta(K+1)\|
	\leq&\prod\limits_{i\in\mathscr{D}}(\rho^{i})^{N_{0i}}  \prod\limits_{i\in\mathcal{S}_{1}} (\omega^{i})^{T_{i}(0,K)}\prod\limits_{i\in\mathcal{S}_{2}} (\frac{1}{\gamma^{i}})^{T_{i}(0,K)} \|\zeta(0)\|\\
	\leq&\prod\limits_{i\in\mathscr{D}}(\rho^{i})^{N_{0i}}  \prod\limits_{i\in\mathscr{D}}(\delta^{i})^{T_{i}(0,K)} \|\zeta(0)\|\\
	\leq&\prod\limits_{i\in\mathscr{D}}(\rho^{i})^{N_{0i}} \delta^{(K+1)} \|\zeta(0)\|, ~\delta=\max\limits_{i}\delta^{i}.
	\end{aligned}
	\end{equation*}
	The stability of the considered system is guaranteed.
	
	We now show the scenario of $\zeta(0)\in \mathscr{C}_{2}$. Likewise, one attains
	\begin{equation*}\label{eq40}
	\begin{aligned}
	\|\zeta(K+1)\|
	\leq&\prod\limits_{i\in \Omega_{1}}\| (A_{\sigma(k_{i+1})})^{\hat{k}_{i}}\|_{\mathscr{C}_{2}} \prod\limits_{i\in \Omega_{2}}\| (A_{\sigma(k_{i+1})})^{\hat{k}_{i}}\|_{\mathscr{C}_{2}}
	\|\zeta(0)\|\\
	\leq&\prod\limits_{i\in \Omega_{1}}\| (A_{\sigma(k_{i+1})})^{\hat{k}_{i}}\|_{\mathscr{H}_{i}^{\ddag}} \prod\limits_{i\in \Omega_{2}}\| (A_{\sigma(k_{i+1})})^{\hat{k}_{i}}\|_{\mathscr{C}_{2}}
	\|\zeta(0)\|\\
	\leq&\prod\limits_{i\in\mathscr{D}}(\rho^{i})^{N_{0i}}\prod\limits_{i\in\mathcal{S}_{1}} (\rho^{i}(\lambda^{i})^{\mathbb{N}_{i}})^{\frac{T_{i}(0,K)}{\mathbb{N}_{i}}}  \prod\limits_{i\in\mathcal{S}_{2}} (\rho^{i})^{\frac{T_{i}(0,K)}{\mathbb{N}_{i}}}\|\zeta(0)\|.\\
	\end{aligned}
	\end{equation*}
	Analogously, the conditions in $(\ref{eq32})$ and $(\ref{eq33})$ are modified by $
	\rho^{i}(\lambda^{i}\gamma^{i})^{\mathbb{N}_{i}}\leq 1, ~i\in \mathcal{S}_{1},~\frac{\rho^{i}}{(\omega^{i})^{\mathbb{N}_{i}}}\leq 1, i\in \mathcal{S}_{2}$.
	Similarly, $(\ref{eq35})$ indicates that
	\begin{equation*}\label{eq41}
	\begin{aligned}
	\|\zeta(K+1)\|
	\leq&\prod\limits_{i\in\mathscr{D}}(\rho^{i})^{N_{0i}}\prod\limits_{i\in\mathcal{S}_{1}} (\frac{1}{\gamma^{i}})^{T_{i}(0,K)}   \prod\limits_{i\in\mathcal{S}_{2}} (\omega^{i})^{T_{i}(0,K)}\|\zeta(0)\|\\
	\leq&\prod\limits_{i\in\mathscr{D}}(\rho^{i})^{N_{0i}} \prod\limits_{i\in\mathscr{D}}(\delta^{i})^{T_{i}(0,K)}\|\zeta(0)\|\\
	\leq&\prod\limits_{i\in\mathscr{D}}(\rho^{i})^{N_{0i}} \delta^{(K+1)} \|\zeta(0)\|, ~\delta=\max\limits_{i}\delta^{i}.\\
	\end{aligned}
	\end{equation*}
	
Finally, we argue the special case for $\zeta(0)\in \overline{\mathscr{C}_{1}\bigcup \mathscr{C}_{2}}$. By using the fact that $\mathscr{C}_{1}\bigcap \mathscr{C}_{2}=\emptyset$, it follows
	that there exist two initial conditions, i.e., $\zeta_{1}(0)\in \mathscr{C}_{1}$, $\zeta_{2}(0)\in \mathscr{C}_{2}$, such that
	\begin{equation}\label{eq42}
	\begin{aligned}
	\zeta(0)=\zeta_{1}(0)+\zeta_{2}(0).
	\end{aligned}
	\end{equation}
Incorporating $(\ref{eq37})$ with $(\ref{eq42})$, it hence yields
	\begin{equation*}\label{eq43}
	\begin{aligned}
	\zeta(K+1)
	=&~(A_{\sigma(k_{N_{\sigma}(0,K)})})^{\hat{k}_{N_{\sigma}(0,K)}}\cdots(A_{\sigma(k_{i+1})})^{\hat{k}_{i}}\cdots (A_{\sigma(k_{1})})^{\hat{k}_{0}}\zeta(0)\\
	=&~(A_{\sigma(k_{N_{\sigma}(0,K)})})^{\hat{k}_{N_{\sigma}(0,K)}}\cdots(A_{\sigma(k_{i+1})})^{\hat{k}_{i}}
	\cdots (A_{\sigma(k_{1})})^{\hat{k}_{0}}\zeta_{1}(0)\\
	&+(A_{\sigma(k_{N_{\sigma}(0,K)})})^{\hat{k}_{N_{\sigma}(0,K)}}\cdots(A_{\sigma(k_{i+1})})^{\hat{k}_{i}}\cdots (A_{\sigma(k_{1})})^{\hat{k}_{0}}\zeta_{2}(0)\\
	=&~\zeta_{1}(K+1)+\zeta_{2}(K+1)
	\end{aligned}
	\end{equation*}
	where $\zeta_{1}(K+1)$ and $\zeta_{2}(K+1)$ are originated from the initial values $\zeta_{1}(0)$ and $\zeta_{2}(0)$, respectively. Recalling the foregoing involved arguments, we can draw the same conclusion as the above two cases.
	
Based on the aforementioned discussions and error equation $(\ref{eq16})$, we conclude that the coordination problem in $(\ref{eq1})$ with $(\ref{eq1u})$ could be solved with the designed switching rule, provided that the requirements in Lemma $\ref{le9}$ hold.
\end{proof}

A significant property revealed in Theorem $\ref{th2}$ lies in the simple fact that the coordination behavior of the agents can be guaranteed if \emph{TDADT} with respect to each communication topology is preserved. When compared with the developed method in \cite{han2013cluster}, we can get rid of the limitation on the computation of Hajnal diameter, which may generally be a huge work or even intractable for the large-scale networks. The last but not least, the condition $(\ref{eq36})$ is vital for solving the coordination problem.

\begin{rmk}
Another novel feature of Theorem $\ref{th2}$ in contrast to the existing works \cite{jadbabaie2003coordination,ren2005consensus}, is that it provides a new perspective in dealing with the coordination problem subject to both changing communication environment and antagonistic information. A potential way guaranteeing the feasibility of the switching signal design in $(\ref{eq32})$-$(\ref{eq35})$ could be formulated as follows: those communication topologies $\mathcal{G}_{i}$, which contain a directed spanning tree, for some $i\in\{1,2,...,\mathbbm{n}\}$ should be implemented with more execution time; while those communication topologies that have several clusters, should be carried out as little time as possible.\QEDA
\end{rmk}

\begin{figure}
	\centering
	\includegraphics[width=3.45in,height=2.35in]{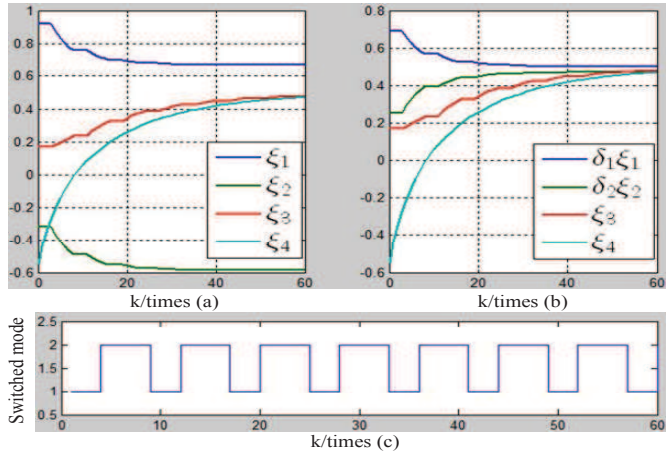}\\
	\caption{(a) The state trajectories of the agents; (b) The scaling states of the agents with the changing topologies; (c) The switched mode of the system.}
	\label{fig21}
\end{figure}

For the sake of improving the state of the art, let us revisit Example $\ref{exa3}$ to support the validity of Theorem $\ref{th2}$.

\begin{exa}(Continuation of Example $\ref{exa3}$)\label{exa4}
It is easy to obtain that $\lambda^{1}=0.85$ and $\lambda^{2}=0.9427$. Choosing $\gamma^{1}=1.01$, $\gamma^{2}=1.03$, $\mathbb{N}_{1}=3$ and $\mathbb{N}_{2}=5$. We then check condition $(\ref{eq33})$ by $\rho^{1}(\lambda^{1}\gamma^{1})^{\mathbb{N}_{1}}=0.6327$ and $\rho^{2}(\lambda^{2}\gamma^{2})^{\mathbb{N}_{2}}=0.8998$, which implies that $(\ref{eq33})$ are desirable. Suppose that $(\ref{eq32})$ are strict equalities. One can obtain that $\omega^{1}=1$ and $\omega^{2}=1$. It is vulnerable to know that the conditions $(\ref{eq34})$ and $(\ref{eq35})$ are satisfied. According to Theorem $\ref{th2}$ and Definition $\ref{de1}$, the coordination problem with the dynamical changing topologies depicted in Fig. $\ref{fig2}$ is solved. The trajectories, scaling states and the switching modes are presented in Fig. $\ref{fig21}$ where the initial values of the autonomous agents are randomly selected from the interval $[-1,1]$. \QEDA
\end{exa}

In the sequel, we put our attention into a simple case where the changing communication topologies contain directed spanning tree.
\begin{cor}\label{cor3}
For the communication graphs with directed spanning tree, coordination in system $(\ref{eq1})$ with $(\ref{eq1u})$ can be guaranteed for any switching signal $\sigma(k)$ satisfying
	$\rho^{i}(\lambda^{i}\gamma^{i})^{\mathbb{N}_{i}}\leq 1$ with $\gamma^{i}>1$, for $i\in \mathscr{D}$.
\end{cor}


Likewise, the scenario where the interactions among rooted agents are absent is concretely taken into account in a switching communication setting.


Evidently, matrix $\mathcal {M}(k)$ attains form by
\begin{equation*}
\mathcal {M}_{\sigma(k)}=\begin{bmatrix}
\mathcal {O}_{M^{2}}&\mathcal {O}_{M(N-M)}\\
\widehat{\mathcal{L}}_{\sigma(k)}&\mathcal{L}_{3,\sigma(k)}
\end{bmatrix}
\end{equation*}
with $\widehat{\mathcal{L}}_{\sigma(k)}= (\mathcal{L}_{2}\mathcal {D})_{\sigma(k)}$.
Consequently, one has the following condiment.
\begin{pro}\label{pro2}
	Suppose $\mathcal{L}_{3,\sigma(k)}\mathcal{L}^{*}_{\sigma(k)}=\Theta$ and $\mathcal{L}^{*}_{\sigma(k)}\mathcal{L}_{3,\sigma(k)}=\Theta$. It follows that
	\begin{equation}\label{eq46}
	\begin{aligned}
	\Theta\widehat{\mathcal{L}}_{\sigma(k)}=\widehat{\mathcal{L}}_{\sigma(k)}
	\end{aligned}
	\end{equation}
	where $\mathcal{L}^{*}_{\sigma(k)}$ is the generalized inverse of the matrix $\mathcal{L}_{3,\sigma(k)}$ at the $k$th instant with
	$0\leq\hbar\leq N-M$ and $\Theta ={\rm diag}(0,...,0,I_{\hbar})$.
\end{pro}

The proof of Proposition $\ref{pro2}$ is drawn in {\scshape Appendix} ${\rm \ref{app3}}$.

It is noted that only the fixed topology is considered in \cite{liu2012necessary}.
Besides, the containment control problem under the switching setting is explicitly argued in \cite{notarstefano2011containment}.
Each changing graph is required to have a directed spanning forest over the time evolution in the aforementioned literature. According to Proposition $\ref{pro2}$, this hypothesis is able to be further relaxed.

However, the united directed spanning tree
in $\mathcal {G}_{i}$ may be loss of preservation over some updating intervals a.


With the help of the aforementioned preparations, we introduce a new variable
\begin{equation*}\label{eq52}
\begin{aligned}
x(k)=\xi_{N-M}(k)+\mathcal{L}^{*}_{\sigma(k)}\widehat{\mathcal{L}}_{\sigma(k)}\xi_{M}(k)
\end{aligned}
\end{equation*}
where $\xi_{M}(k)\triangleq[\xi_{1}(k),...,\xi_{M}(k)]^{\prime}$ and $\xi_{N-M}(k)\triangleq[\xi_{M+1}(k),...,\xi_{N}(k)]^{\prime}$ for simplicity.
As a consequence, dynamic equation for variable $x(k)$ is coherently written by
\begin{equation}\label{eq53}
\begin{aligned}
x(k+1)=(I-\epsilon\mathcal{L}_{3,\sigma(k)})x(k), ~k\in \mathbb{Z}.
\end{aligned}
\end{equation}

Similar to Lemma $\ref{le9}$, we need a precondition for directed spanning forest before presenting the main result.
\begin{lem}\label{le10}
The union graph of a collection of the changing topologies has a directed spanning forest if there exist two sets $\mathcal{S}_{1}$ and $\mathcal{S}_{2}$, both of which are contained in $\mathscr{D}$ and $\mathcal{S}_{1}\bigcup\mathcal{S}_{2}=\mathscr{D}$ is desirable, such that $\mathscr{C}_{1}=\sum\limits_{i\in\mathcal{S}_{1}}\mathscr{H}_{i}^{\dag}$ and $\mathscr{C}_{2}=\bigcap\limits _{ i\in \mathcal{S}_{1}} \mathscr{H}_{i}^{\ddag}$ are $(I-\epsilon\mathcal{L}_{3,i})$-invariant sets, and $\mathscr{C}_{1 }\subseteq\bigcap\limits_{i\in\mathcal{S}_{2}}\mathscr{H}_{i}^{\ddag}$.
\end{lem}
\begin{proof}
	The proof of Lemma $\ref{le10}$ is similar to that in Lemma $\ref{le9}$, and is hence omitted.
\end{proof}

At this stage, a result in the absence of the reciprocity among the leaders is reported.
\begin{cor}\label{th3}
	Given scalars $\omega_{i}>0$ and $\gamma_{i}>1$. Assume that the conditions in Lemma $\ref{le10}$ are fulfilled.
	Then coordination for $(\ref{eq1})$ with changing topologies is achievable, if the switching signal $\sigma(k)$ is chosen such that
\begin{subequations}  
\begin{align}  
&	\frac{\rho^{i}}{(\omega_{i})^{\mathbb{N}_{i}}}\leq 1,~i\in \mathscr{D}\\
&\rho^{i}(\lambda_{i}\gamma_{i})^{\mathbb{N}_{i}}\leq 1, ~i\in \mathscr{D}\\      
&\prod\limits_{i\in\mathcal{S}_{1}}(\omega_{i})^{T_{i}(0,K)} \prod\limits_{i\in\mathcal{S}_{2}}(\frac{1}{\gamma_{i}})^{T_{i}(0,K)}
\leq\prod\limits_{i\in\mathscr{D}}(\widehat{\delta}_{i})^{T_{i}(0,K)}\\
&\prod\limits_{i\in\mathcal{S}_{1}}(\frac{1}{\gamma_{i}})^{T_{i}(0,K)}
\prod\limits_{i\in\mathcal{S}_{2}}(\omega_{i})^{T_{i}(0,K)} \leq\prod\limits_{i\in\mathscr{D}}(\widehat{\delta}_{i})^{T_{i}(0,K)}
\end{align}
\end{subequations}
where $\lambda_{i}$ corresponds to the matrix $(I-\epsilon\mathcal{L}_{3,i})$ in $(\ref{eq53})$, and is defined in a similar manner to $(\ref{eq31})$. $\rho^{i}$ has the same meaning as that in Theorem $\ref{th2}$, and $\widehat{\delta}_{i} \in (0,1)$.
\end{cor}
\begin{proof}
	The proof is analogous to that of Theorem $\ref{th2}$, and is hence omitted.
\end{proof}


\subsection{Further Discussion}
The principle purpose in subsequent part is to discuss the coordination problem with fixed communication topology. Moreover, some connections with the existing literature shall also be elaborated.

As for the case where interacting topology among agents is fixed, one obtains a simple result.
\begin{cor}\label{th1}
Suppose that communication graph attains a directed spanning tree. Then coordination problem formulated in $(\ref{eq1})$ can be exponentially addressed if and only if $(\ref{eq16})$ is exponentially stable. Moreover, the trajectories of the non-rooted agents will be aggregated to a common quantity which is relevant to the weighted gains, scaling parameters, topology and the initial information of the leaders. In addition, the final states of arbitrarily pairwise rooted agents are proportional to their corresponding scaling parameters, i.e., $\lim_{k\to \infty}\delta_{i}\xi_{i}(k)=\lim_{k\to \infty}\delta_{j}\xi_{j}(k)$ for all $i,j \in \hat{\mathbbm{I}}_{N}$.
\end{cor}
\begin{proof}
In virtue of Lemma $\ref{le4}$, the first statement is obvious. The exponential convergence rate is concretely computed by $\|A\|\leq\max | \lambda(A)|$. In the light of directed spanning tree condition, Lemma $\ref{le5}$ reveals that matrix $\mathscr{L}$ has a unique $1$ eigenvalue and the remaining eigenvalues are strictly located in the unit disk. According to Lemma $\ref{le4}$, it follows that $\lambda\triangleq\max | \lambda(A)|<1$. One hence obtains that $\lambda=\exp(\ln(\lambda))$
	with $-\ln(\lambda)>0$, which is the exponential convergence rate.
	
We now show the remaining statement.
This argument is trivial according to the results developed in Lemmas $\ref{le1}$ and $\ref{le4}$. It is worth noting that the left and the right
	eigenvectors associated with $1$ eigenvalue for $\mathscr{L}$ in $(\ref{eq2})$ can be selected by $\phi=[\frac{p_{1}}{\varrho_{1}},...,\frac{p_{M}}{\varrho_{M}},0,...,0]$ and
	$\varphi=[\frac{1}{\delta_{1}},...,\frac{1}{\delta_{M}},1,...,1]^{\prime}$ in combining Lemma $\ref{le1}$ with $(\ref{eq1})$.
	Hence, the item $\lim\limits_{k\rightarrow\infty}\xi_{i}(k)$ for $i\in \hat{\mathbbm{I}}^{c}_{N}$ relies on the weighted gains, scaling parameters, communication topology and the initial conditions of the rooted agents by orthogonalizing the vectors $\phi$ and $\varphi$.
	In virtue of $(\ref{eq8})$,
	the proof hence follows according to Definition $\ref{de1}$.
\end{proof}

Note that an interesting counterexample given in \cite{meng2016behavior} shows that the mutual requirement of the directed spanning tree is not sufficient for preserving the bipartite consensus, when the cooperative formulation developed in the aforementioned reference is employed. In the sequel, we shall show that the coordination problem investigated in current paper actually involves some advantages over the existing results.
\begin{figure}
	\centering
	\includegraphics[width=4.85in,height=1.5in]{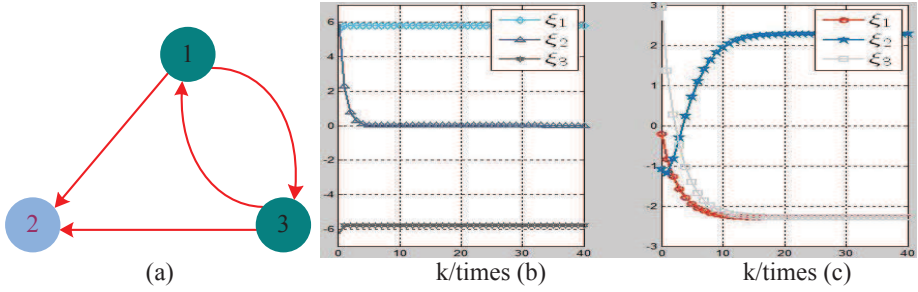}
	\caption{(a) A communication graph; The state trajectories with the algorithm in \cite{meng2016behavior} (subgraph (b)) and in this paper (subgraph (c)) where the initial conditions are randomly selected from $[-10,10]$.}
	\label{fig3}
\end{figure}
\begin{exa}\label{exa2}
	We now consider a communication graph which is described in Fig. $\ref{fig3}$(a), where the red edges imply that the communication among three agents is antagonistic. Notice that the communication topology is quasi-strongly connected according to \cite{proskurnikov2016opinion,meng2016behavior}.
	
	Let us revisit the updating rule proposed in \cite{meng2016behavior}. The states of the interactive agents are updated according to
	$\xi(k+1)=\mathcal {A}\xi(k)$, where
	\begin{equation*}
	\mathcal {A}=\frac{1}{6}\begin{bmatrix}
	3&0&-3\\
	-2&2&-2\\
	-3&0&3
	\end{bmatrix}
	\end{equation*}
	with $\sum^{3}_{j=1}|a_{ij}|=1$ for all $i=1,2,3$. Unfortunately, the designed coordination algorithm in \cite{meng2016behavior} fails to guarantee the bipartite consensus (cf. \cite{proskurnikov2016opinion}), and the trajectories of the agents are given in Fig. $\ref{fig3}$(b). We further show that there does not exist a group of eligible composition for $\varrho_{i}$ and $\delta_{i}$ based on the cooperative framework in \cite{meng2016behavior}.
	
We now compute the underlying equality: $\mathcal {A}=\mathscr{L}=I-\epsilon\mathcal {M}$. By tedious computations, one has that $\epsilon=\frac{1}{3},~\varrho_{1}=\frac{3}{2},~\varrho_{3}=\frac{3}{2},~\delta_{1}=-1,~\delta_{3}=-1$,
which in turn implies that the conclusion in Theorem $\ref{th5}$ is violated leading to the failure of the coordination. The ultimate reason is principally due to the additional restriction, that is, $\sum^{3}_{j=1}|a_{ij}|=1$ on the interaction topology. Alternatively, in terms of the coordination strategy $(\ref{eq1u})$, we can assert that the considered coordination problem can be addressed according to the conclusion of Theorem $\ref{th5}$, and the simulation results are represented in Fig. $\ref{fig3}$(c) with $\epsilon=0.2$, $\delta_{1}=-1$, $\delta_{3}=-1$, $\varrho_{1}= 1$, and $\varrho_{3}=-2.5$ where $a_{ij}=1$ or $0$. Hence, the proposed coordination framework $(\ref{eq1})$ may take some merits compared with the existing work.\QEDA
\end{exa}


The containment control problem is extensively investigated with a common hypothesis that the leaders are isolated (e.g., see \cite{liu2012necessary}). For such a situation, we can also derive an analogue result by virtue of the coordination algorithm $(\ref{eq1u})$ with some slight modification. 

\begin{cor}\label{cor1}
For the case where there is no any interaction among leaders, that is, $\mathcal{L}_{1}=\mathcal {O}_{M^{2}}$. Then, coordination among agents is achieved if and only if graph $\mathcal{G}$ contains a directed spanning forest. In addition, the state variables for the followers converge to $-\mathcal{L}^{-1}_{3}\mathcal{L}_{2}\mathcal {D}\xi_{M}(0)$ eventually where $\xi_{M}(0)$ is the initial states with respect to the first $M$ agents.
\end{cor}
\begin{proof}
According to \cite[ Lemma $4$]{liu2012necessary} and the foregoing discussions, coordination problem in the absence of interactions among the leaders is desirable.
	
We now calculate the final converging values for the followers. By
	computations, one has that
	$\lim\limits_{k\rightarrow\infty}\xi(k)=\lim\limits_{k\rightarrow\infty}\mathscr{L}^{k}\xi(0)$ where
	\begin{equation*}\label{eq22}
	\begin{aligned}
	\mathscr{L}^{k}=
	\begin{bmatrix}
	I_{M^{2}} &\mathcal {O}_{M(N-M)}\\
	\Psi&\Phi^{k}
	\end{bmatrix}
	\end{aligned}
	\end{equation*}
	with $\Phi=I_{(N-M)^{2}}-\epsilon\mathcal{L}_{3}$ and $\Psi=-\epsilon \Bigg(\sum\limits^{k-1}_{i=0}\Phi^{i}\Bigg)\mathcal{L}_{2}\mathcal {D}$ for $k\geq1$. Notice that the eigenvalues of the matrix $\Phi$ are totally contained in the unit disk. As a result, $\lim\limits_{k\rightarrow\infty}\Phi^{k}=\mathcal {O}_{(N-M)^{2}}$. Thereby, the item $\sum\limits^{k-1}_{i=0}\Phi^{i}$ converges and its limit exists. Hence, it directly yields that $
	\lim\limits_{k\rightarrow\infty}\sum\limits^{k-1}_{i=0}\Phi^{i}=\frac{1}{\epsilon}\mathcal{L}^{-1}_{3}$.
	By utilizing the foregoing arguments, the final state value for each node is explicitly described by
	\begin{equation*}
	\lim\limits_{k\rightarrow\infty}\xi(k)=
	\begin{bmatrix}
	\xi_{M}(0) \\
	-\mathcal{L}^{-1}_{3}\mathcal{L}_{2}\mathcal {D} \xi_{M}(0)
	\end{bmatrix}
	\end{equation*}
	This concludes the proof.
\end{proof}

\begin{rmk}
In view of Corollary $\ref{cor1}$, we incorporate the result presented in \cite{liu2012necessary} as a special case. This in turn means that the designed coordination algorithm in this paper is prevailing over some of the existing coordination algorithms in the literature to some extent. Nevertheless, we are unable to declare that the followers are contained in a convex hull spanned by the leaders' initial values. The reason is that the final values of the followers depend on not only the initial conditions of the leaders, but also the scaling parameters which normally do not share a common amplitude or a mutual sign.\QEDA
\end{rmk}

Notice that each leader is either rewarding or hostile to these nodes who can receive its information according to the coordination algorithm $(\ref{eq1u})$. The formulation of interest is somewhat motivated by a simple fact: in a social network, it is feasible that an individual is treated as a liar by the rest of the participating members if the individual tells a lie, even just once. The case that a leader may send beneficial and antagonistic information to the other nodes could coexist at the same updating time. We believe that the obtained results in this paper could be extended to the aforementioned scenario, and this interesting problem will be investigated in the near future.

\subsection{Connection with Altafini's Model}\label{Main 34}
The bipartite consensus problem is the main core of the Altafini's model \cite[Equation $(3)$]{altafini2013consensus}. First of all, let's review the Altafini's model
\begin{equation}\label{FOeq001}
\begin{aligned}
\dot{\xi}(t)=-L\xi(t)
\end{aligned}
\end{equation}
where the adjacency matrix is denoted by $\mathbb{A}^{\star}=[a^{\star}_{ij}]_{N^{2}}$ with $a^{\star}_{ij}\in \mathbb{R}$. $L=[l^{\star}_{ij}]_{N^{2}}$ with $l^{\star}_{ii}=\sum^{N}_{j=1,j\neq i}| a^{\star}_{ij}|$ and $l^{\star}_{ij}=-a^{\star}_{ij}$. Obviously, the relationship between $\mathbb{A}$ and $\mathbb{A}^{\star}$ is $a_{ij}=| a^{\star}_{ij}|$, i.e., they differ merely in the weights of some edges.

\begin{pro}\label{pro5}
	Suppose that the communication graph is strongly connected, and Altafini's model $(\ref{FOeq001})$ achieves the bipartite consensus (cf. \cite[Definition 1]{altafini2013consensus}). Then there exist a group of $(\varrho_{i},\delta_{i})$, by which $(\ref{eq1})$ also achieves the same bipartite consensus as $(\ref{FOeq001})$. Moreover, $(\ref{eq1})$ can achieve a general bipartite consensus in the sense of $|\xi_{i}(t)|=|\xi_{j}(t)|$ provided that $(\varrho_{i},\delta_{i})$ satisfies Theorem $\ref{th5}$ in addition to $| \delta_{i}|=| \delta_{j}|$ is desirable.
\end{pro}

The remarkable difference between the Altafini's model $(\ref{FOeq001})$ and $(\ref{eq1})$ is: Altafini's model uses the weights of the edges to characterizes the antagonistic information; While $(\ref{eq1})$ uses the nodes to characterize the antagonistic information.

\begin{proof}
	For the sake of consistency with the Altafini's model, we take into account the continuous version of $(\ref{eq1})$
	\begin{equation}\label{FOeq002}
	\begin{aligned}
	\dot{\xi}(t)=-\mathbb{D}\mathcal{L}\mathcal {D}\xi(t)
	\end{aligned}
	\end{equation}
	According to \cite[Lemma 2]{altafini2013consensus}, $(\ref{FOeq001})$ solves the bipartite consensus problem if and only if there exists a diagonal matrix $D$, where the $i$th diagonal element of $D$ is $1$ or $-1$, such that $\mathcal{L}=DLD$ holds. For such a scenario, it suffices to let $\mathbb{D}=D$ and $\mathcal {D}=D$ hold. For that reason, it arrives at $L=\mathbb{D}\mathcal{L}\mathcal {D}$ (since $D^{-1}=D$).
	Hence, the conclusion is trivial according to Definition $\ref{de1}$ and Theorem $\ref{th5}$.
	
	We now show the second part. By virtue of Theorem $\ref{th5}$, $\mathbb{D}\mathcal{L}\mathcal {D}$ has a simple zero eigenvalue and the nonzero eigenvalues have positive real parts. By Theorem $\ref{th1}$, $(\ref{eq1})$ can solve the bipartite consensus problem. This completes the proof.
\end{proof}


Built upon Altafini's model, interval bipartite consensus problem was elaborated \cite[Defintion $1$]{meng2016interval}. For system $(\ref{FOeq001})$, it solves the interval bipartite consensus problem if $\lim_{t\rightarrow \infty} | \xi_{i}(t)|=\varpi$ on condition that agent $i$ is a root, otherwise $\lim_{t\rightarrow \infty} | \xi_{i}(t)|\leq\varpi$ where $\varpi>0$ is a constant.

\begin{thm}\label{th6}
	Suppose that the underlying graph contains a directed spanning tree. Then there exist a group of $(\varrho_{i},\delta_{i})$, by which $(\ref{eq1})$ also solves the interval bipartite consensus as $(\ref{FOeq001})$.
\end{thm}

\begin{proof}
	Suppose that $(\ref{FOeq001})$ guarantees interval bipartite consensus in the sense that $\lim_{t\rightarrow \infty}  \xi(t)\triangleq \widetilde{\xi}(\infty)=[\widetilde{\xi}_{1}(\infty),...,\widetilde{\xi}_{M}(\infty),\widetilde{\xi}_{M+1}(\infty),...,
	\widetilde{\xi}_{N}(\infty)]^{\prime}\neq \mathcal {O}_{N1}$. The proof is divided into two cases.
	
	Case 1: $\widetilde{\xi}_{i}(\infty)\neq 0$ for $i\in  \hat{\mathbbm{I}}^{c}_{N}$. For system $(\ref{FOeq002})$, we bring in a new agent, indexed by $0$, with $\xi_{0}(0)\neq0$ (virtual leader) yielding an augment system
	\begin{equation*}\label{FOeq003}
	\begin{aligned}
	\dot{\xi}(t)=
	\begin{bmatrix}
	0&\mathcal {O}_{1N}\\
	-\mathbb{D}\delta_{0}\ell&-\mathbb{D}\mathcal{L}\mathcal {D}
	\end{bmatrix}\xi(t)
	\end{aligned}
	\end{equation*}
	where $\delta_{0}\neq 0$, and $\ell\neq \mathcal {O}_{N1}$ quantifies the influence of the virtual leader imposing on the remaining agents. Let
	\begin{equation*}\label{FOeq004}
	\begin{aligned}
	\delta_{i}=\delta_{0}\iota_{i}, ~~i\in\mathbb{I}_{N}
	\end{aligned}
	\end{equation*}
	where $\xi_{0}(0)=\iota_{i}\widetilde{\xi}_{i}(\infty)$ and $\iota_{i}\neq0$. Consequently, one gets that $\lim_{t\rightarrow \infty}\xi_{i}(t)=\widetilde{\xi}_{i}(\infty)$. By virtue of Definition $\ref{de1}$ and Theorem $\ref{th5}$, there exists a diagonal matrix $\mathbb{D}$ such that system $(\ref{FOeq002})$ converges to $\widetilde{\xi}(\infty)$.
	
	Case 2: there exists an index $i$ satisfying $\xi_{i}(\infty)=0$ where $i\in\hat{\mathbbm{I}}^{c}_{N}$ (see also Fig. $\ref{fig3}$(b) for instance). We then delete the index since the agent has a trivial state. The argument is the same as that in Case $1$ (if there are multiple indices, the conclusion also holds). This completes the proof.
\end{proof}

In the context of interval bipartite consensus, rooted agents always achieve the bipartite consensus if $\varpi>0$ holds (cf. \cite[Thm. $5$]{meng2016interval}), which is consistent with \cite[Lemma 2]{altafini2013consensus} where the underlying graph is strongly connected. In some situation, we just require that agents belonging to set $\hat{\mathbbm{I}}^{c}_{N}$ are bounded by $\varpi$ since they are more arduous to determine in general (cf. \cite{meng2016interval}). For such a scenario, we merely need to choose $\max_{i}|\delta_{i}|\leq |\delta_{j}| $ where $i\in \hat{\mathbbm{I}}_{N}$ and $j\in \hat{\mathbbm{I}}^{c}_{N}$.
By doing so, we can always transform the Altafini's model (edge-based algorithm) into the formulation (node-based algorithm) promoted in this paper. Unfortunately, the converse may be not true since it is more arduous to select a collection of the edges whose wights are assigned to be negative except for the constraint regarding the digon sign-symmetry. This is because the number of the edges in a graph is far more larger than the nodes whenever the graph is larger-scale. More interesting, $(\ref{eq1})$ could still guarantee the bipartite consensus, even if there are some pairs of $(\varrho_{i},\delta_{i})$ enjoying the property of ${\rm sgn}(\varrho_{i})\neq{\rm sgn}(\delta_{i})$.

\section{Conclusion}\label{sec5}
This paper has studied coordination problem for time-varying antagonistic networks. It is shown that interactions among agents and the description of antagonistic information can be quantified by a fully decoupled manner. Based on the proposed coordination algorithm, it has guaranteed the existence of weighted gains, established connection among scaling parameter, weighted gain and communication topology as well as locations of the eigenvalues in system matrix. Topology-dependent average time condition has further provided to devise the underlying changing topologies, relaxing the dependence on infinite product of nonnegative matrices and Lyapunov-based techniques. Some discussion and comparison have been carried out to support the effectiveness of the proposed coordination algorithm.

\section{Appendices}

\subsection{Proof of Lemma $\ref{wle1}$}\label{app521}
\begin{proof}
	We use the inductive method to prove it. Proposition $\ref{pro4}$ indicates that all the $s$th ($1\leq s \leq M-1$) order principal minors of $\mathcal{L}(k)$ are not equal to zero. The conclusion is evident for $M=2$. Suppose that any $(n-1)$th order principal minor of $\mathcal{L}(k)$ is positive for $n>2$. We will show the case of $n$. Let $\Phi_{n}(k)=(\phi_{ij}(k))_{n^{2}}$ being any $n$th order principal minor of $\mathcal{L}(k)$. Since the underlying graph induced by $\mathcal{L}(k)$ is strongly connected, there exists at least an index $i$ such that $\phi_{ii}(k)$ is strictly diagonal dominant, and we denote by $i=1$ without loss of generality. By Schur determinantal formulae, it has that
	\begin{equation*}
	\begin{aligned}
	{\rm det}(\Phi_{n}(k))=&{\rm det}
	\begin{bmatrix}
	\phi_{11}(k)&\phi_{1}(k)\\
	\phi_{2}(k)&\Phi_{n-1}(k)
	\end{bmatrix}\\
	=&{\rm det}(\Phi_{n-1}(k)){\rm det}(\phi_{11}(k)-\phi_{1}(k)  \Phi^{-1}_{n-1} (k)\phi_{2}(k))
	\end{aligned}
	\end{equation*}
	Notice that
	$\begin{bmatrix}
	\phi_{2}(k)&\Phi_{n-1}(k)
	\end{bmatrix}\textbf{1}\geq \mathcal {O}_{(n-1)1}$ with $\textbf{1}=[1,...,1]^{\prime}$. It implies that $\textbf{1}\geq -\Phi^{-1}_{n-1} (k)\phi_{2}(k)\geq \mathcal {O}_{(n-1)1}$. Therefore, the fact that $\phi_{11}(k)-\phi_{1}(k)  \Phi^{-1}_{n-1}(k) \phi_{2}(k)>0$ directly leads to ${\rm det}(\Phi_{n}(k))>0$. This completes the proof.
\end{proof}

\subsection{Proof of Lemma $\ref{le3}$}\label{app5}
\begin{proof}    
Apparently, it suffices to prove that the equality $P\mathscr{L}(k)QP=P\mathscr{L}(k)$ holds. In essence, it is enough to examine the entries in the $N$th column with the aid of the block matrix. Moreover, one has
\begin{equation*}\label{eq13}
	\begin{aligned}
	QP=\begin{bmatrix}
	I_{(N-1)^{2}}&\nu\\
	\mathcal {O}_{ 1(N-1)}&0
	\end{bmatrix}
	\end{aligned}
\end{equation*}
where $\nu=-[\frac{1}{\delta_{1}},...,\frac{1}{\delta_{M}},1,...,1]^{\prime}$. Let us start by checking the characterization of the entries in matrix $P\mathscr{L}(k)=[e_{ij}(k)]_{(N-1)N}\triangleq[e^{\prime}_{1}(k),...,e^{\prime}_{N-1}(k)]^{\prime}$ as follows:
\begin{equation*}\label{eq14}
\begin{aligned}
	&e_{i(i+1)}(k)=\delta_{i+1}(-1-\epsilon(\varrho_{i}\delta_{i}l_{i(i+1)}(k)-\varrho_{i+1}\delta_{i+1}l_{(i+1)^{2}}(k))),~i\in \hat{\mathbbm{I}}_{N}\setminus \{M\}\\
	&e_{i^{2}}(k)=\delta_{i}(1-\epsilon(\varrho_{i}\delta_{i}l_{i^{2}}(k)-\varrho_{i+1}\delta_{i+1}l_{(i+1)i}(k))),~i\in \hat{\mathbbm{I}}_{N}\setminus \{M\}\\
	&e_{ij}(k)=\epsilon\delta_{j}(\varrho_{i+1}\delta_{i+1}l_{(i+1)j}(k)-\varrho_{i}\delta_{i}l_{ij}(k)),~i\in \hat{\mathbbm{I}}_{N}\setminus \{M\} ,~j\in \hat{\mathbbm{I}}_{N},~j\neq i,~i+1 \\
	&e_{ij}(k)=0,~i\in \hat{\mathbbm{I}}_{N}, ~j\in\hat{\mathbbm{I}}^{c}_{N}\\
	&e_{M^{2}}(k)=\delta_{M}(1-\epsilon(\delta_{M}\varrho_{M} l_{M^{2}}(k)-l_{(M+1)M}(k)))\\
	&e_{Mj}(k)=\epsilon\delta_{j}(l_{(M+1)j}(k)-\varrho_{M}\delta_{M}l_{Mj}(k)),~j\in \hat{\mathbbm{I}}_{N}\setminus \{M\}\\
	&e_{M(M+1)}(k)=-1+\epsilon l_{(M+1)^{2}}(k)\\
	&e_{Mj}(k)=\epsilon l_{(M+1)j}(k),~j=M+2,...,N\\
	&e_{ij}(k)=\epsilon\delta_{j}(l_{(i+1)j}(k)-l_{ij}(k)),~i\in \mathbbm{I}^{\star}_{N} \triangleq\{ M+1,...,N-1\},~j\in \hat{\mathbbm{I}}_{N}\\
	&e_{i^{2}}(k)=1-\epsilon(l_{i^{2}}(k)-l_{(i+1)i}(k)),~i\in \mathbbm{I}^{\star}_{N} \\
	&e_{i(i+1)}(k)=-1+\epsilon(l_{(i+1)^{2}}(k)-l_{i(i+1)}(k)),~i\in \mathbbm{I}^{\star}_{N}\\
	&e_{ij}(k)=\epsilon(l_{(i+1)j}(k)-l_{ij}(k)),~i\in \mathbbm{I}^{\star}_{N} ,~j\in \mathbbm{I}^{c}_{N}
	\end{aligned}
\end{equation*}
By using the property of the Laplacian matrix together with lengthy calculations, one has that $e^{\prime}_{i}(k)\mathbbm{h}=e_{iN}(k)$ for $i \in \mathbbm{I}_{N}\setminus \{N\}$,
where $\mathbbm{h}=[\nu^{\prime},0]^{\prime}$. It follows immediately that $P\mathscr{L}(k)\mathbbm{h}=[0,...,0,e_{(M+1)N}(k),...,e_{(N-1)N}(k)]^{\prime}$.
\end{proof}

\subsection{{Proof of Lemma $\ref{le5}$}\label{app6}}
\begin{proof}
	The characteristic equation of the involved system equipped with $\mathscr{L}(k)$ is depicted by
	\begin{equation*}
	\begin{aligned}
	\mathscr{F}_{k}(\lambda)=&~{\rm det}(\lambda I-\mathscr{L}(k))\\
	=&~{\rm det}\bigg((\lambda-1) I+\epsilon\mathcal {M}(k)\bigg)\\
	=&~{\rm det}\bigg(\epsilon\bigg[\frac{(\lambda-1)}{\epsilon} I+\mathcal {M}(k)\bigg]\bigg)\\
	=&~{\rm det}(\epsilon I){\rm det}\bigg(\frac{(\lambda-1)}{\epsilon} I+\mathcal {M}(k)\bigg)=0
	\end{aligned}
	\end{equation*}
	It further gets that
	\begin{equation*}\label{eqw6}
	\begin{aligned}
	0=&~{\rm det}\bigg(\frac{(\lambda-1)}{\epsilon} I+\mathcal {M}(k)\bigg)\\
	\triangleq &~{\rm det}(\lambda^{\star} I+\mathcal {M}(k))\triangleq\mathscr{F}_{k}(\lambda^{\star})
	\end{aligned}
	\end{equation*}
	Evidently, $\mathscr{F}_{k}(\lambda^{\star})$ represents the characteristic equation of the underlying system with matrix $-\mathcal {M}(k)$. By Theorem $\ref{th5}$, $\mathcal {M}(k)$ has a simple zero eigenvalue and the nonzero eigenvalues have positive real parts.
	Whence, the connection between the eigenvalues of the matrices $\mathscr{L}(k)$ and $-\mathcal {M}(k)$ is reported by
	\begin{equation*}\label{eqw7}
	\begin{aligned}
	\lambda=1+\epsilon\lambda^{\star}
	\end{aligned}
	\end{equation*}
	Therefore, one can choose by
	\begin{equation*}\label{eqw8}
	\begin{aligned}
	0< \epsilon< \min _{i\in \mathbbm{I}_{N} \backslash \{j\},\lambda^{\star}_{j}=0}\bigg\{\frac{-2\lambda^{\star}_{1i}}{(\lambda^{\star}_{1i})^{2}+(\lambda^{\star}_{2i})^{2}}\bigg\}
	\end{aligned}
	\end{equation*}
	where $\lambda^{\star}_{i}=\lambda^{\star}_{1i}+\mathbbm{i}\lambda^{\star}_{2i}$, by which $| \lambda |<1$ always holds whenever $\lambda^{\star}\neq0$. This ends the proof.
\end{proof}

\subsection{Proof of Lemma $\ref{le4}$}\label{app1}
\begin{proof}
Here we just discuss the case where graph $\mathcal{G}(k)$ preserves a directed spanning tree since the other scenarios could be investigated with minor modifications. Computing the eigenvalues of $\mathscr{L}(k)$ by
	\begin{equation*}\label{eq59}
	\begin{aligned}
	{\rm det}(\lambda I-\mathscr{L}(k))=&\left| \lambda I-\mathscr{L}(k)\right|\\
	=&\left|\begin{matrix}
	\kappa_{1^{2}} (k)&\epsilon\varrho_{1}\delta_{2}l_{12}(k)& \cdots& 0\\
	\epsilon\varrho_{2}\delta_{1}l_{21}(k) & \kappa_{2^{2}}(k)& \cdots& 0\\
	\vdots & \vdots &\ddots&\vdots\\
	\epsilon\delta_{1}l_{N1}(k)& \epsilon\delta_{2}l_{N2}(k)&\cdots&\kappa_{N^{2}}(k)
	\end{matrix}\right|\\
	=&\widehat{\varsigma}\left|\begin{matrix}
	\widetilde{\kappa}_{1^{2}} (k)&\epsilon\varrho_{1} l_{12}(k)& \cdots& 0\\
	\epsilon\varrho_{2} l_{21} (k)& \widetilde{\kappa}_{2^{2}}(k)& \cdots& 0\\
	\vdots & \vdots &\ddots&\vdots\\
	\epsilon l_{N1}(k)& \epsilon l_{N2}(k)&\cdots&\kappa_{N^{2}}(k)
	\end{matrix}\right|\\
	=&\varsigma\left|\begin{matrix}
	\widehat{\kappa}_{1^{2}} (k)&\epsilon l_{12}(k)& \cdots& 0\\
	\epsilon  l_{21} (k)& \widehat{\kappa}_{2^{2}}(k)& \cdots& 0\\
	\vdots & \vdots &\ddots&\vdots\\
	\epsilon l_{N1}(k)& \epsilon l_{N2}(k)&\cdots&\kappa_{N^{2}}(k)
	\end{matrix}\right|\\
	=&\left|\begin{matrix}
	\kappa_{1^{2}} (k)&\epsilon\varrho_{1}\delta_{1}l_{12}(k)& \cdots& 0\\
	\epsilon\varrho_{2}\delta_{2}l_{21} (k)& \kappa_{2^{2}}(k)& \cdots& 0\\
	\vdots & \vdots &\ddots&\vdots\\
	\epsilon l_{N1}(k)& \epsilon l_{N2}(k)&\cdots&\kappa_{N^{2}}(k)
	\end{matrix}\right|\\
	\end{aligned}
	\end{equation*}
	with $\kappa_{i^{2}}(k)=(\lambda-1)+\epsilon\varrho_{i}\delta_{i}l_{i^{2}}(k)$, $\widetilde{\kappa}_{i^{2}}(k)=\frac{(\lambda-1)}{\delta_{i}}+\epsilon\varrho_{i}l_{i^{2}}(k)$, $\widehat{\kappa}_{i^{2}}(k)=\frac{(\lambda-1)}{\varrho_{i}\delta_{i}}+\epsilon l_{i^{2}}(k)$
	for $i\in \hat{\mathbbm{I}}_{N}$, and $\kappa_{i^{2}}(k)=(\lambda-1)+\epsilon l_{i^{2}}(k)$ for $i \in \hat{\mathbbm{I}}^{c}_{N}$. Besides, $\widehat{\varsigma}=\prod\limits^{M}_{i=1}\delta_{i}$ and $\varsigma=\prod\limits^{M}_{i=1}\varrho_{i}\delta_{i}$. Now,
	we perform the transformation step by step: $\textbf{1)}$ Subtract $(i+1)$th row from $i$th row for $i<N$; $\textbf{2)}$ Add the first $i$ columns to $(i+1)$th for all $1\leq i<N$. It follows that
	\begin{equation*}
	\begin{aligned}
	{\rm det}(\lambda I-\mathscr{L}(k))=&\left| \lambda I-\mathscr{L}(k)\right|\\
	=&\left|\begin{matrix}
	\lambda-\mathbbm{a}_{1^{2}}(k) &-\mathbbm{a}_{12}(k)& \cdots&0 &0\\
	-\mathbbm{a}_{21} (k)& \lambda-\mathbbm{a}_{2^{2}}(k)&\cdots & 0&0\\
	\vdots & \vdots &\ddots&\vdots&\vdots\\
	-\mathbbm{a}_{(N-1)1}(k) &-\mathbbm{a}_{(N-1)2} (k)&\cdots &\lambda-\mathbbm{a}_{(N-1)^{2}}(k)&0\\
	\epsilon l_{N1}(k)& \epsilon(l_{N1}(k)+l_{N2}(k))&\cdots&\epsilon\sum\limits^{N-1}_{j=1}l_{Nj}(k)&\lambda-1
	\end{matrix}\right|\\
	=&(\lambda-1){\rm det}(\lambda I-A(k)).
	\end{aligned}
	\end{equation*}
	The above argument suggests that matrix $A(k)$ possesses the same eigenvalues as matrix $\mathscr{L}(k)$ except for an eigenvalue $1$.
	In addition, matrix $\mathscr{L}(k)$ has exactly an eigenvalue at $1$ and the remaining eigenvalues are contained in the unit disk if directed spanning tree is preserved; the converse is also true (cf. Lemma $\ref{le5}$).
	This completes the proof.
\end{proof}

\subsection{Proof of Lemma $\ref{le7}$}\label{app7}
\begin{proof}
Obviously, both $\mathscr{H}^{\dag}$ and $\mathscr{H}^{\ddag}$ are $A(k)$-invariant subspaces since
	$A(k)x=x \in \mathscr{H}^{\dag}$ for $x\in \mathscr{H}^{\dag}$, and $A(k)y=\lambda y \in \mathscr{H}^{\ddag}$ for
	$y\in \mathscr{H}^{\ddag}$ in terms of Lemma $\ref{le6}$.
	According to Lemma $\ref{le4}$,
	we can simply choose the orthogonal matrix $U$ of the form
	\begin{equation}\label{eq28}
	\begin{aligned}
	U=[u_{1},...,u_{d},u_{d+1},...,u_{N-1}]
	\end{aligned}
	\end{equation}
	where $0\leq d\leq N-1$. Therefore, $\mathscr{H}^{\dag}$ and $\mathscr{H}^{\ddag}$ are, respectively, spanned by the base vectors $[u_{1},...,u_{d}]$ and
	$[u_{d+1},...,u_{N-1}]$. Hence, one obtains
	\begin{equation}\label{eq29}
	\begin{aligned}
	UA(k)U^{-1}={\rm diag}(I_{d},\textbf{A}(k))
	\end{aligned}
	\end{equation}
	where $I_{d}$ is a $d$-dimensional identity matrix, and $\textbf{A}(k)$ is upper triangular with $|\lambda(\textbf{A}(k))|<1$. From $(\ref{eq28})$ and $(\ref{eq29})$, we have
	\begin{equation*}\label{eq30}
	\begin{aligned}
	\|A(k)\| _{\mathscr{H}^{\dag}}\leq&~ \|U\|_{\mathscr{H}^{\dag}} \| U^{-1} \|_{\mathscr{H}^{\dag}}\leq \rho\\
	\| A(k)\| _{\mathscr{H}^{\ddag}}\leq& ~\| U\|_{\mathscr{H}^{\ddag}} \| U ^{-1}\|_{\mathscr{H}^{\ddag}}\| \textbf{A}(k)\|\leq \rho \lambda
	\end{aligned}
	\end{equation*}
	where $\rho\triangleq \max \{ \| U\|_{\mathscr{H}^{\dag}} \| U ^{-1}\|_{\mathscr{H}^{\dag}},
	\| U\|_{\mathscr{H}^{\ddag}} \| U ^{-1}\|_{\mathscr{H}^{\ddag}}\}$.
	This ends the proof.
\end{proof}

\subsection{Proof of Lemma $\ref{le9}$}\label{app8}
\begin{proof}
	It is sufficient to show that there exists at least one index $1\leq\hbar \leq \mathbbm{n}$ such that $\lambda^{s}(A_{\hbar})\neq 1$ for all $s=1,...,(N-1)$ according to Lemma $\ref{le12}$. The analysis is divided into three cases.
	
	\textbf{Case} $1$: $\mathscr{C}_{1}=\emptyset$. With the help of the facts that $\mathscr{C}_{1}\bigcap\mathscr{C}_{2}=\emptyset$, $\mathcal{S}_{1}\subseteq\mathscr{D}$ and $\mathcal{S}_{2}\subseteq\mathscr{D}$ are nonempty (it is trivial when $\mathcal{S}_{1}$ or $\mathcal{S}_{2}$ is empty), one obtains that the magnitude of the eigenvalues on $A_{i}$ for $i\in \mathcal{S}_{1}$ is strictly less than $1$. Thereby, graph $\mathcal{G}_{i}$ possesses a directed spanning tree by employing Lemma $\ref{le4}$.
	
	\textbf{Case} $2$: $\mathscr{C}_{2}=\emptyset$. This situation is similar to \textbf{Case} $1$ because of $(\ref{eq36})$.
	
	\textbf{Case} $3$: Both $\mathscr{C}_{1}$ and $\mathscr{C}_{2}$ are nonempty. In a nutshell, we label the linearly independent eigenvectors belonging to $\mathscr{C}_{1}$ and $\mathscr{C}_{2}$ by $\{v^{1}_{1},...,v^{1}_{s}    \}$ and $\{v^{2}_{s+1},...,v^{2}_{N-1}    \}$ for some integer $s$ satisfying $1\leq s<(N-1)$. Evidently, $\mathscr{C}_{1}= {\rm span}\{v^{1}_{1},...,v^{1}_{s}    \}$ and $\mathscr{C}_{2}= {\rm span}\{v^{2}_{s+1},...,v^{2}_{N-1}    \}$. On the other hand, condition $(\ref{eq36})$ and $\mathscr{C}_{2}$ imply that there are $(N-1)$ linearly independent eigenvectors corresponding to these eigenvalues located in the unit disk. It in turn indicates that there exist at least $(N-1)$ eigenvalues whose magnitudes are smaller than $1$. Note that the case, the eigenvalue multiplicity is more than one, could be discussed in the same way because of $(\ref{eq36})$. In the light of Lemma $\ref{le12}$, we can conclude that the union graph preserves a directed spanning tree.
\end{proof}

\subsection{Proof of Proposition $\ref{pro2}$}\label{app3}
\begin{proof}
	The analysis for Proposition $\ref{pro2}$ is divided into four cases.
	
	\textbf{Case} $1$: The graph $\mathcal{G}(k)$ contains a directed spanning forest. In such a case, it is clear that $
	\mathcal{L}^{*}_{\sigma(k)}=\mathcal{L}^{-1}_{3,\sigma(k)}$.
	Then the conclusion is obvious. On the other hand, for the worst scenario where all nodes are isolated, one obtains that $\widehat{\mathcal{L}}_{\sigma(k)}=\mathcal {O}_{(N-M)M}$ and $\mathcal{L}_{3,\sigma(k)}=\mathcal {O}_{(N-M)^{2}}$. Evidently, $(\ref{eq46})$ always holds.
	
	\textbf{Case} $2$: Suppose that $\mathcal{G}(k)$ does not have a directed spanning forest. For simplicity, assume that merely the $(M+1)$th node is isolated at the $k$th time instant since the other situations can be discussed similarly. It hence follows that the entries of $(M+1)$th row in matrix $\mathcal {M}_{\sigma(k)}$ are equivalent to zero. Thereby, $
	\mathcal{L}_{3,\sigma(k)}\mathcal{L}^{*}_{\sigma(k)}=\Theta$ and $\mathcal{L}^{*}_{\sigma(k)}\mathcal{L}_{3,\sigma(k)}=\Theta$, where $\Theta={\rm diag}(0,I_{(N-M-1)^{2}})$.
	With the help of the fact that the first row in $\widehat{\mathcal{L}}_{\sigma(k)}$ is a zero row vector, it is straightforward that $(\ref{eq46})$ is true.
	
	\textbf{Case} $3$: Assume that merely the $(M+1)$th node belongs to a subgraph $\widetilde{\mathcal{G}}(k)$ where no path connects a root of $\mathcal{G}(k)$ to nodes in $\widetilde{\mathcal{G}}(k)$, while $\mathcal{G}(k)\setminus\widetilde{\mathcal{G}}(k)$ contains a directed spanning forest. Without loss of generality, it assumes that only the $(M+1)$th node and $(M+2)$th node are contained in $\widetilde{\mathcal{G}}(k)$, and the $(M+1)$th node is the parent of the $(M+2)$th node, not vice versa. Evidently, the $(M+1)$th row of the matrix $\mathcal {M}_{\sigma(k)}$ is zero. The same conclusion can be drawn as that in \textbf{Case} $2$.

	\textbf{Case} $4$: We now turn to the case where the $(M+1)$th node and $(M+2)$th node are the parents of each other. For this case, matrix $\mathcal{L}_{3,\sigma(k)}$ enjoys the following form
	\begin{equation*}\label{eq49}
	\begin{aligned}
	\mathcal{L}_{3,\sigma(k)}=
	\begin{bmatrix}
	l_{(M+1)^{2}}& -l_{(M+1)^{2}}&\mathcal {O}_{1(N-M-2)}\\
	-l_{(M+2)^{2}}&l_{(M+2)^{2}}&\mathcal {O}_{1(N-M-2)}\\
	\star&\star&\widetilde{\mathcal{L}}_{\sigma(k)}
	\end{bmatrix}
	\end{aligned}
	\end{equation*}
	where $l_{(M+1)^{2}}$ and $l_{(M+2)^{2}}$ are some positive scalars. The asterisk $``\star"$
	stands for the induced term by symmetry. Matrix $\widetilde{\mathcal{L}}_{\sigma(k)}$ is nonsingular at the $k$th instant. For this situation, there exists a nonsingular matrix $\mathbb{P}={\rm diag}(\widehat{\mathbb{P}} ,I)$ such that
	\begin{equation*}\label{eq50}
	\begin{aligned}
	\mathbb{P}\mathcal{L}_{3,\sigma(k)}\mathbb{P}^{-1}=
	\begin{bmatrix}
	0& 0&\mathcal {O}_{1(N-M-2)}\\
	0&l^{0}_{(M+2)^{2}}&\mathcal {O}_{1(N-M-2)}\\
	\star&\star&\widetilde{\mathcal{L}}_{\sigma(k)}
	\end{bmatrix}
	\end{aligned}
	\end{equation*}
	where $l^{0}_{(M+2)^{2}}\neq 0$.
	The rest of the argument goes back to \textbf{Case} $2$ and \textbf{Case} $3$.
	Notice that the first one or two rows in the matrix $\mathcal{L}_{3,\sigma(k)}$ are equivalent to zero. This indicates that the rank of the matrix $\widehat{\mathcal{L}}_{\sigma(k)}$ is no more than $\hbar$. More importantly, note that if the $i$th diagonal entry in $\Theta$ equals to zero, then the $i$th row in $\widehat{\mathcal{L}}_{\sigma(k)}$ must be the zero vector relying on the aforementioned arguments.
	Hence, $\Theta\widehat{\mathcal{L}}_{\sigma(k)}=\widehat{\mathcal{L}}_{\sigma(k)}$ is always desirable.
	This completes the proof.
\end{proof}




\bibliographystyle{unsrt}   
\bibliography{mybib}       

\end{document}